\documentclass[a4paper,11pt,english]{article}
\usepackage[utf8]{inputenc}
\usepackage{babel,amsmath,amssymb,amsthm,enumitem,xcolor}
\usepackage[sort&compress,numbers,square]{natbib}
\usepackage[colorlinks,allcolors=blue]{hyperref}
\usepackage[margin=1.25in]{geometry}
\usepackage[small,bf]{titlesec}

\titlelabel{\thetitle.\hspace{0.5em}}

\DeclareMathOperator{\E}{E}
\DeclareMathOperator{\I}{I}
\renewcommand{\P}{\mathrm{P}}
\newcommand{\Q}{\mathrm{Q}}
\newcommand{\bnu}{\bar\nu}
\newcommand{\EE}{\mathcal{E}}
\newcommand{\D}{\mathcal{D}}
\newcommand{\e}{\mathcal{E}}
\newcommand{\B}{\mathcal{B}}
\newcommand{\F}{\mathcal{F}}
\newcommand{\PP}{\mathcal{P}}
\newcommand{\FF}{\mathbb{F}}
\newcommand{\R}{\mathbb{R}}
\newcommand{\RMpp}{\mathbb{R}^M_{++}}
\newcommand{\RNP}{{\mathbb{R}^N_+}}
\newcommand{\s}{{(s)}}
\newcommand{\bl}{\boldsymbol{l}}
\newcommand{\bL}{\boldsymbol{L}}
\newcommand{\cadlag}{{c\`adl\`ag}}

\renewcommand{\hat}{\widehat}
\renewcommand{\tilde}{\widetilde}
\renewcommand{\epsilon}{\varepsilon}
\newcommand{\cint}{\boldsymbol{\cdot}}
\newcommand{\cd}{\boldsymbol{\cdot}}
\newcommand{\dd}[2]{\frac{d{#1}}{d{#2}}}

\newtheorem{theorem}{Theorem}
\newtheorem{lemma}{Lemma}
\newtheorem{proposition}{Proposition}
\theoremstyle{definition}
\newtheorem{definition}{Definition}
\newtheorem*{remark*}{Remark}

\title{A continuous-time asset market game with short-lived assets}
\author{Mikhail Zhitlukhin\thanks{Steklov Mathematical Institute of the
Russian Academy of Sciences. 8 Gubkina St., Moscow, Russia. Email:
mikhailzh@mi-ras.ru. The research was supported by the Russian Science
Foundation, project no. 18-71-10097.}}
\date{30 August 2020}

\begin{document}
\maketitle

\begin{abstract}
We consider a continuous-time game-theoretic model of an investment market
with short-lived assets and endogenous asset prices. The first goal of the
paper is to formulate a stochastic equation which determines wealth
processes of investors and to provide conditions for the existence of its
solution. The second goal is to show that there exists a strategy such that
the logarithm of the relative wealth of an investor who uses it is a
submartingale regardless of the strategies of the other investors, and the
relative wealth of any other essentially different strategy vanishes
asymptotically. This strategy can be considered as an optimal growth
portfolio in the model.

\medskip \textit{Keywords:} asset market game, relative growth optimal
strategy, martingale convergence, evolutionary finance.

\medskip
\noindent
\emph{MSC 2010:} 91A25, 91B55. \emph{JEL Classification:} C73, G11.
\end{abstract}

\section{Introduction}
This paper proposes a dynamic game-theoretic model of an investment market
-- an \emph{asset market game} -- and study strategies that allow an
investor to achieve faster growth of wealth compared to rival market
participants. The model provides an outlook on growth optimal portfolios
different from the well-known theory in a single-investor setting, which
originated with \citet{Kelly56} and \citet{Breiman61} (see also
\cite{KaratzasKardaras07,PlatenHeath06,AlgoetCover88} for a modern
exposition of the subject). Our results belong to the strand of research on
\emph{evolutionary finance} -- the field which studies financial markets
from a point of view of evolutionary dynamics and investigates properties of
investment strategies like survival, extinction, dominance, and how they
affect the structure of a market. Reviews of recent results in this
direction can be found in, e.g., \cite{Holtfort19,EvstigneevHens+16}. While
the majority of models in evolutionary finance are discrete-time, the
novelty and one of the goals of this paper consists in developing a
continuous-time model.

The model considered here describes a market consisting of several assets
and investors. The assets yield random payoffs which are divided between the
investors proportionally to the number of shares of each asset held by an
investor. One feature of the model, which makes it different from the
classical optimal growth theory, is that the asset prices are determined
endogenously by a short-run equilibrium of supply and demand and depend on
the investors' strategies. As a result, the investor's wealth depends not
only on their own strategies and realized assets' payoffs but also on
strategies of the other investors in the market.

One of the main results of the paper is a proof of the existence of a
strategy such that the logarithm of the relative wealth of an investor who
uses it is a submartingale, regardless of the strategies used by the other
investors (by the relative wealth we mean the share of wealth of one
investor in the total wealth of the market). In particular, we do not assume
that investors are necessarily rational in the sense that their actions can
be described as solutions of some well-posed optimization problems, and they
need not be aware of strategies used by their rivals. Remarkably, the
optimal strategy that we find needs only to know the current total market
wealth and the probability distribution of future payoffs, but does not
require the knowledge of the other investors' strategies or their individual
wealth. Such a strategy can be attractive for possible applications, since
quantitative information about individual market agents is always scarce.
Besides this submartingale property, the strategy has other good
characteristics (similar to those of growth optimal strategies in
single-investor models), which will be also investigated in the paper.

The importance of these results consists in more than just the fact of
existence of a ``good'' strategy -- they also allow to describe the
asymptotic structure of a market, i.e.\ the asymptotic distribution of
wealth, the asset prices, and the representative strategy of all investors.
We prove that if at least one investor uses the optimal strategy, then the
relative wealth of the other investors, who use essentially different
strategies, will vanish asymptotically and those investors will have a
vanishing impact on the market. Consequently, the market becomes determined
by investors who use the optimal strategy or strategies which are
asymptotically close to it. Results of this type are well-known in the
literature, beginning with the seminal paper of \citet{BlumeEasley92}. It is
worth noting that many of them do not employ standard game-theoretic
concepts of a solution of a game, e.g.\ the Nash equilibrium, but rather
seek for ``unbeatable'' or ``winning'' strategies (though sometimes one can
show that an optimal strategy provides a Nash equilibrium for a particular
payoff function). In some cases, this fact can be regarded as an advantage
as it allows to model arbitrary behavior of market agents and does not rely
on unobservable characteristics like individual utilities or beliefs.

In this paper we deal with a simplified market model and consider a market
with only \emph{short-lived} assets. Such assets can be viewed as short-term
investment projects rather than, e.g., common stock -- they are traded at
time $t$, yield payoffs at the ``next infinitesimal'' moment of time, and
then get replaced with new assets. Short-lived assets have no liquidation
value, so investors can get profit (or loss) only from receiving asset
payoffs and paying for buying new assets. Despite being a simplification of
real stock markets, models with short-lived assets have been widely studied
in the literature because they are more amenable to mathematical analysis
and ideas developed for them may be transferred to advanced models (see a
discussion in \cite{EvstigneevHens+16}).

This paper is tightly connected with paper \cite{DrokinZhitlukhin20}, which
considers the same model in discrete time. Regarding mathematical methods,
both of the papers are based on the approach proposed
by~\citet{AmirEvstigneev+13}, which directly shows (in discrete time) that
the logarithm of wealth of an investor who uses the optimal strategy is a
submartingale regardless of the other investors' strategies (see also the
paper of~\citet{AmirEvstigneev+11} where similar but technically more
involved ideas were used for a model with long-lived assets). Then, using
martingale convergence theorems, we can obtain results about the asymptotic
structure of a market. This martingale approach is more general compared to
methods used in earlier works, which were based on assumptions that payoff
sequences and/or strategies are stationary (as in, e.g.,
\cite{HensSchenkHoppe05,EvstigneevHens+06}). An essential difference of our
model and the model of \cite{AmirEvstigneev+13} (in addition to that we
consider a continuous-time model) is that Amir et al. assume that market
agents spend their whole wealth for purchasing assets in each time period,
so the total market wealth is always equal the most recent total payoff of
the assets. On the other hand, our model includes a risk-free asset (cash or
a bank account with zero interest rate) that can be used by investors to
store capital. This leads to more complicated wealth dynamics, but is
necessary for consideration of a continuous-time model, where asset payoffs
can be infinitesimal but yielded in a continuous way. Moreover, adding the
possibility to store capital in cash opens interesting questions about the
asymptotics of the total market wealth, which do not arise in models where
the whole wealth is spend for purchasing assets with random payoffs. For
example, as was observed in \cite{DrokinZhitlukhin20}, greater uncertainty
in asset payoffs may result in faster growth of investors' wealth -- a fact
which at first may seem counter-intuitive. In the present paper, we consider
similar questions for the continuous-time model.

In evolutionary finance, there are few models with continuous time. One can
mention the papers
of~\citet{PalczewskiSchenkHoppe10-JME,PalczewskiSchenkHoppe10-JEDC}, in
which a continuous-time model with long-lived assets is constructed. The
paper~\cite{PalczewskiSchenkHoppe10-JEDC} proves that the model can be
obtained as a limit of discrete-time models, and
\cite{PalczewskiSchenkHoppe10-JME} investigates questions of survival of
investments strategies in it. However, their results are obtained only for
time-independent strategies and under the assumption that cumulative
dividend processes are pathwise absolutely continuous. In the present paper,
we allow strategies to be time-dependent and asset payoffs to be represented
by arbitrary processes. A~continuous-time model with short-lived assets was
also constructed in \cite{Zhitlukhin18}. An essential limitation of that
paper consists in the assumption that all investors spend the same
proportion of their wealth for purchasing assets which is specified
exogenously. This makes the mathematical analysis of the model considerably
simpler compared to the present paper, both in showing that the wealth
process is well-defined, and in construction of the optimal strategy.

The paper is organized as follows. In Section~\ref{sec-discrete-model}, we
briefly describe a discrete-time model, which helps to explain the main
ideas of the paper. The general model is formulated in
Section~\ref{sec-model}. In Section~\ref{sec-optimal-strategy}, we define
the notion of optimality of a strategy and construct a candidate optimal
strategy. In Section~\ref{sec-main}, we formulate the main results, which
state that this strategy is indeed optimal, and investigate some of its
properties. Section~\ref{sec-proofs} contains the proofs of the results. In
the \hyperref[appendix]{appendix}, we formulate and prove several auxiliary
facts about the Lebesgue decomposition and Lebesgue derivatives of
non-decreasing random functions which are used in the paper.

\section{Preliminary consideration: a discrete-time model}
\label{sec-discrete-model}
In this section, we describe the main ideas of the paper using a simple
model with discrete time which avoids technical details of continuous time.
Based on the discrete-time model, in Section~\ref{sec-model} a general
continuous-time model will be formulated. The model presented here is a
slightly simplified version of the model from \cite{DrokinZhitlukhin20}.

Let $(\Omega,\F,\P)$ be a probability space with a filtration
$\FF=(\F_t)_{t=0}^\infty$. The model includes $M\ge 2$ investors and $N\ge
1$ assets which yield non-negative random payoffs at moments of time
$t=1,2,\ldots$ The assets live for one period: they are purchased by the
investors at time $t$, yield payoffs at $t+1$, and then the cycle repeats.
The asset prices are determined endogenously by a short-run equilibrium of
supply and demand. The supply of each asset is normalized to 1, and the
demand depends on actions of the investors. The payoffs are specified in an
exogenous way, i.e.\ do not depend on the investor's actions. Each investor
receives a part of a payoff yielded by an asset which is proportional to the
owned share of this asset.

The asset payoffs are specified by random sequences $A_t^n\ge 0$ adapted to
the filtration. The wealth of investor $m$ is described by an adapted random
sequence $Y_t^m\ge 0$. The initial wealth $Y_0^m$ of each investor is
non-random and strictly positive. The wealth $Y_t^m$ at subsequent moments
of time $t\ge 1$ is determined by the investors' strategies and the asset
payoffs.

A strategy of investor $m$ is a plan according to which this investor
allocates the available budget $Y_t^m$ towards a purchase of assets. Such an
allocation is specified by a sequence of vectors
$l_t^m=(l_t^{m,1},\ldots,l_t^{m,N})$, where $l_t^{m,n}$ is a budget
allocated towards a purchase of asset $n$ at time $t-1$. At each moment of
time, the vectors $l_t^m$ are selected by the investors simultaneously and
independently, so the model represents a simultaneous-move $N$-person
dynamic game, and $l_t^m$ represent the investors' actions. These actions
may depend on a random outcome $\omega$ and current and past wealth of the
investors, so we define a strategy $\bl^m$ of investor $m$ as a sequence of
$\F_{t-1}\otimes\B(\R^{tM}_+)$-measurable functions
\[
\bl_t^{m}(\omega, y_0,\ldots,y_{t-1})\colon \Omega\times \R^{tM}_+\to
\R_+^N,\qquad t=1,2,\ldots.
\]
(We will use boldface letters to distinguish between strategies and their
realizations, see below.) The arguments $y_s=(y_s^1,\ldots,y_s^M)\in
\R^M_+$, $s\le {t-1}$, correspond to the wealth of the investors at the past
moments of time. It is assumed that short sales are not allowed, so
$l_t^{m,n}\ge 0$, and it is not possible to borrow money, so $\sum_n
l_t^{m,n} \le y_{t-1}^m$. The amount of wealth $y_{t-1}^m - \sum_n
l_t^{m,n}$ is held in cash and carried forward to the next time period.

After selection of investment budgets $l_t^m$ by the investors, the
equilibrium asset prices $p_{t-1}^n$ are determined from the market clearing
condition that the aggregate demand for each asset is equal to the aggregate
supply, which is assumed to be 1. At time $t-1$, investor $m$ can buy
$x_t^{m,n} = l_t^{m,n} / p_{t-1}^n$ units of asset $n$, so its price at time
$t-1$ should be equal to the total amount of capital invested in this asset,
$p_{t-1}^n = \sum_m l_{t}^{m,n}$. If $\sum_m l_t^{m,n}=0$, i.e. no one buys
asset $n$, we put $p_{t-1}^n=0$ and $x_t^{m,n}=0$ for all $m$.

Thus, investor $m$'s portfolio between moments of time $t-1$ and $t$
consists of $x_t^{m,n}$ units of asset $n$ and $c_t^{m} := y_{t-1}^m-\sum_n
l_t^{m,n}$ units of cash. At a moment of time $t$, the total payoff received
by this investor from the assets in the portfolio is equal to $\sum_n
x_t^{m,n} A_t^n$. In our model, the assets have no liquidation value, so the
budgets used at time $t-1$ for buying assets are not returned to the
investors. Consequently, investor $m$'s wealth is described by the adapted
sequence $Y_t^m$ which is defined by the recursive relation
\begin{equation}
Y_{t}^m(\omega) = Y_{t-1}^m(\omega)-\sum_{n=1}^N l_{t}^{m,n}(\omega) + \sum_{n=1}^N
\frac{l_t^{m,n}(\omega)}{\sum_k
l_{t}^{k,n}(\omega)} A_t^n(\omega),\qquad
t\ge 1,
\label{capital-equation-discrete}
\end{equation}
where $l_t^{m,n}(\omega) = \bl_t^{m,n}(\omega, Y_0,Y_1(\omega),\ldots,
Y_{t-1}(\omega))$ are the realizations of the investors' strategies, with
$0/0=0$ in the right-hand side of \eqref{capital-equation-discrete}.

Note that the investors' actions precede the asset prices, so they first
``announce'' the budgets they plan to allocate for buying the assets, and
then the prices are adjusted to clear the market. This modeling approach is
analogous to market games of Shapley--Shubik type. Its justification and
details can be found in, e.g., \citet{ShapleyShubik77}. Also, one can see
that the asset prices do not enter equation
\eqref{capital-equation-discrete}. They are needed to provide a financial
interpretation of the equation, but we will not work with them directly.

We will be mainly interested in relative wealth of investors. For investor
$m$, we define the relative wealth as the adapted sequence
\[
r_t^m = \frac{Y_t^m}{\sum_k Y_t^k}.
\]
Our goal will be to identify a strategy such that the relative wealth of an
investor who uses it grows in the following sense: for any strategies of the
other investors and any initial wealth, the sequence $\ln r_t^m$ is a
submartingale (as a consequence, $r_t^m$ will be a submartingale as well).
Such a strategy will exhibit several asymptotic optimality properties, which
we will consider in Sections \ref{sec-optimal-strategy} and \ref{sec-main}.

\section{The general model}
\label{sec-model}
In order to formulate a continuous-time counterpart of
equation~\eqref{capital-equation-discrete}, observe that it can be written
in the following form:
\begin{equation}
\Delta Y_t^m(\omega) =  -\sum_{n=1}^N \Delta L_t^{m,n}(\omega) +
\sum_{n=1}^N \frac{\Delta L_t^{m,n}(\omega)}
{\sum_k \Delta L_t^{k,n}(\omega)} \Delta X_t^n(\omega),
\label{capital-equation-rewritten}
\end{equation}
where 
\[
L_t^{m,n}(\omega) = \sum_{s=1}^t l_s^{m,n}(\omega), \qquad X_t^n(\omega) =
\sum_{s=1}^t A_s^n(\omega)
\]
are, respectively, the process of the cumulative wealth invested by investor
$m$ in asset~$n$ and the cumulative payoff process of asset $n$. The symbol
$\Delta$ denotes a one-step increment, e.g. $\Delta Y_t^m =
Y^m_t-Y^m_{t-1}$.

The form of equation \eqref{capital-equation-rewritten} suggests that an
analogous model with continuous time can be obtained by considering
continuous-time processes $X_t,Y_t,L_t$ and replacing one-step
increments with infinitesimal increments, e.g.\ $\Delta X_t$ with $d X_t$.
Our next goal will be to define such a model properly. The model
we are about to formulate includes the above discrete-time model as a
particular case, but we do not investigate convergence of the discrete-time
model to the general model.

\medskip\noindent
\textbf{Notation.} We will work on a filtered probability
space $(\Omega,\F,\FF,\P)$ with a continuous-time filtration
$\FF=(\F_t)_{t\in\R_+}$ satisfying the usual assumptions. By $\PP$ we will
denote the predictable $\sigma$-algebra on $\Omega\times\R_+$.

As usual, equalities and inequalities for random variables are assumed to
hold with probability one. For random processes, an equality $X=Y$ is
understood to hold up to $\P$-indistinguishability, i.e.\ $\P(\exists\, t :
X_t\neq Y_t) = 0$; in the same way we treat inequalities. Pathwise
properties (continuity, monotonicity, etc.) are assumed to hold for all
$\omega$.

For vectors $x,y\in \R^N$, by $xy=\sum_n x^ny^n$ we denote the scalar
product, by $|x|=\sum_n |x^n|$ the $l_1$-norm of a vector, and by $\|x\| =
\sqrt{xx}$ the $l_2$-norm. For a scalar function $f\colon \R\to\R$ and a
vector $x$ the notation $f(x)$ means the application of the function to each
coordinate of the vector, $f(x) = (f(x^1),\ldots, f(x^N))$. If $x\in
\R^{MN}$, we denote by $x^m$ the vector $(x^{m,1}, \ldots,x^{m,N})\in\R^N$
and by $x^{\cd,n}$ the vector $(x^{1,n},\ldots, x^{M,n})\in \R^M$. The
maximum of two numbers $a,b$ is written as $a\vee b$, and the minimum as
$a\wedge b$.

The notation $\xi\cint G_t$ is used for the integral of a process $\xi$ with
respect to a process $G$. In what follows, all the integrators are
non-decreasing \cadlag\ processes, so the integrals are understood in the
pathwise Lebesgue-Stieltjes sense ($f\cint G_t(\omega) = \int_0^t
f_s(\omega) d G_s(\omega)$). If $f, G$ are vector-valued, then $f\cint G_t =
\sum_n f^n \cint G^n_t$.

\subsection{Payoff processes and investment strategies}
There are $N\ge 1$ assets yielding random payoffs which are distributed
between $M\ge 2$ investors. The cumulative payoffs are represented by
exogenous adapted non-decreasing \cadlag\ processes $X_t$ with values in
$\RNP$. Without loss of generality $X_0=0$.

A strategy of investor $m$ is identified with a function $\bL$ which
represents the cumulative wealth invested in each asset and assuming values
in $\RNP$. In order to specify how a strategy may depend on the past history
of the market, let $(D,\D,(\D_t)_{t\ge0})$ denote the filtered measurable
space consisting of the space $D$ of non-negative \cadlag\ functions
$y\colon\R_+\to \R^M_+$, the filtration $\D_t = \sigma(d_u,u\le t)$, where
$d_u$ is the mapping $d_u(y) = y_u$ for $y\in D$, and $\D = \bigvee_{t\ge0}
\D_t$. Elements $y$ of the space $D$ represent possible paths of the wealth
processes of the investors (which are yet to be defined) on the whole time
axis $\R_+$. The wealth of each investor cannot become negative (this
assumption will be imposed on a solution of the wealth equation in the next
section), hence $y$ assume values in $\RNP$.

Let  $(E, \EE, (\EE_t)_{t\ge0})$ be the filtered measurable space with
\[
E = \Omega\times D, \qquad \EE_t = \F_t\otimes \D_t, \qquad \EE =
\bigvee_{t\ge0} \EE_t.
\]
Let $\PP^E$ denote the predictable $\sigma$-algebra on $E\times \R_+$, i.e.
$\PP^E$ is generated by all measurable functions $\xi(\omega,y,t)\colon
E\times \R_+ \to \R$ which are left-continuous in $t$ for any fixed
$(\omega,y)$ and $\EE_t$-measurable for any fixed $t$. In what follows,
functions $\xi(\omega,y,t)$ will be often written as $\xi_t(\omega,y)$, or
$\xi_t(y)$ when omitting $\omega$ does not lead to confusion.

\begin{definition}
A strategy of an investor is a $\PP^E$-measurable function $\bL_t(\omega,y)$
with values in $\R_+^N$ and $\bL_0(\omega,y)=0$, which is non-decreasing and \cadlag\ in~$t$.
\end{definition}

The following lemma will be used further in the construction of the model.

\begin{lemma}
\label{lem-composition}
Let $\bL_t(y)$ be a $\PP^E$-measurable function, and $Y$ an adapted
\cadlag\ process with values in $\R_+^M$. Then the process $L_t(\omega) =
\bL_t(\omega,Y(\omega))$ is predictable ($\PP$-measurable).
\end{lemma}
\begin{proof}
The $\sigma$-algebra $\PP^E$ is generated by sets
$C\times[s,\infty)$, where $s\ge 0$ and $C \in \EE_{s-}$ (as usual,
$\EE_{s-}=\bigvee_{u<t} \EE_u$ and $\EE_{0-} = \EE_0$), see
\cite[\S\,1.2]{LiptserShiryaev89en}.
Hence, approximating $\bL_t(y)$ by simple $\PP^E$-measurable functions, it
is enough to prove the lemma for functions 
\begin{equation}
\bL_t(\omega,y) = \I((\omega,y)\in C)\I(t\ge s), \qquad C\in \EE_{s-},
\; s\ge 0.\label{lem-composition-1}
\end{equation}
Using that $\D_s$ is generated by sets
\[
\{y\in D : y_{s_i} \in B_i,\; i=1\ldots,n\},
\]
where $s_1 < \ldots < s_n\le s$ and $B_i\in \B(\R_+^M)$, one can see that in
\eqref{lem-composition-1} it is enough to consider only sets $C$ of the
form
\[
C = A\times \{y\in D : y_{s_i} \in B_i,\; i=1\ldots,n\},\quad A \in
\F_{s_n},\;  s_1<\ldots< s_n<s.
\]
For such sets, $\I((\omega,Y(\omega)) \in C) = \I(\omega\in A)
\I(Y_{s_i}(\omega) \in B_i,\; i\le n)$ is $\F_{s_n}$-measurable, and, hence,
$\F_{s-}$-measurable. Therefore, $L_t(\omega) = \I((\omega,Y(\omega))\in C)
\I(t\ge s)$ is $\PP$-measurable.
\end{proof}

\subsection{The wealth equation}
The wealth of the investors will be described by a \cadlag\ adapted process $Y$
with values in $\R_+^M$. In this section we state the equation which
defines~$Y$. We will always assume that the initial wealth $Y_0^m$ of each
investor is non-random and strictly positive. The set of vectors $y\in
\R^M_+$ with all strictly positive coordinates will be denoted by $\RMpp$.

Let $X^c$ denote the continuous part of the payoff process $X$, i.e.\ the
non-decreasing process with values in $\R_+^N$ defined as
\[
X_t^c = X_t - \sum_{s\le t} \Delta X_t,
\]
where $\Delta X_s = X_s - X_{s-}$ and for $s=0$ we put $\Delta X_0 = 0$.
Denote by $\mu$ the measure of jumps of $X$ and by $\nu$ its compensator.
Define the predictable scalar process $G$ (the so-called operational time
process) as
\[
G_t = |X_t^c| + (|x|\wedge1)*\nu_t,
\]
where the star denotes integration with respect to $\nu$, i.e.\  for a
measurable function $f(\omega,t,x)$ on $\Omega\times\R_+\times\R^N_+$
\[
f*\nu_t(\omega) = \int_0^t f(\omega,s,x) \nu(\omega,ds,dx).
\]
Note that $X^c$ is not the continuous martingale part of $X$, as is usually
denoted in stochastic calculus. Actually, all the martingales in our paper
will have zero continuous part, so the notation $X^c$ should not lead to
confusion.

Let $H$ be an arbitrary scalar predictable \cadlag\ non-decreasing process
such that $G\ll H$ (i.e. for a.a.\ $\omega$ the measure on $\R_+$ generated
by the function $G_t(\omega)$ is absolutely continuous with respect to the
measure generated by $H_t(\omega)$).

\begin{definition}
\label{def-feasible-profile}
We call a strategy profile $(\bL^1,\ldots,\bL^M)$ and a vector of initial
wealth $y_0\in \RMpp$ \emph{feasible} if there exists a unique (up to
$\P$-indistinguishability) non-negative \cadlag\ adapted process $Y$, called
the \emph{wealth process}, which assumes values in $\R_+^M$ and satisfies
the following conditions:

\begin{enumerate}[leftmargin=*,label=\arabic*),itemsep=0mm,topsep=0mm] 
\item $Y$ solves the \emph{wealth equation}
\begin{equation}
d Y_t^m = - d |L^m_t| + \sum_{n=1}^N \frac{l_t^{m,n}}{|l_t^{\cd,n}|}
d X_t^n,\qquad Y_0^m=y_0^m\label{capital-eq-diff} \\
\end{equation}
for $m=1,\ldots,M$, where $L_t^m(\omega) = \bL_t^m(\omega, Y(\omega))$, and
$l$ is any $\P\otimes H$-version of the $\R_+^{MN}$-valued process of
predictable Lebesgue derivatives (see the \hyperref[appendix]{appendix} for
details on Lebesgue derivatives; the measure $\P\otimes H$ is defined as in
\eqref{A-product-measure} there)
\begin{equation}
l_t^{m,n} = \frac{d L_t^{m,n}}{d H_t};\label{l-def}
\end{equation}

\item if $Y_t^m(\omega) = 0$ or $Y_{t-}^m(\omega) =
0$, then $L_s^m(\omega) = L_{t-}^m(\omega)$ and $Y_s(\omega) = 0$ for all
$s\ge t$.
\end{enumerate}
\end{definition}

When in \eqref{capital-eq-diff} we have $|l_t^{\cd,n}(\omega)| = 0$ for some
$\omega,t,n$, we put $l_t^{m,n}(\omega)/|l_t^{\cd,n}(\omega)| = 0$.
Observe that the derivatives $l$ are well-defined, since if $Y$ is an
adapted \cadlag\ process, then $L^{m,n}$ is a predictable process according
to Lemma~\ref{lem-composition}.

As usual, equation~\eqref{capital-eq-diff} should be understood in the
integral sense (a.s.\
for all $t$):
\begin{equation}
\label{capital-eq-integ}
Y_t^m = Y^m_0 - |L_t^m| + \sum_{n=1}^N \int_0^t\frac{l_s^{m,n}}{|l_s^{\cd,n}|} d X_s^n,
\end{equation}
where the integral is understood as a pathwise Lebesgue--Stieltjes
integral. It is well-defined since the process $X$ is \cadlag\ and
non-decreasing, and the integrand is non-negative and bounded.

Let us clarify that we use Lebesgue derivatives in the wealth equation and
not Radon-Nikodym derivatives (e.g.\ $dL^{m,n}_t / d |L^{\cdot,n}_t|$) for
two reasons. First, this allows to differentiate with respect to a process
$H$ not depending on the solution of the equation, which is yet to be found.
Second, it is natural to require that the solution should not depend on what
particular version of the derivatives is used. This is so if $G\ll H$ (see
Proposition~\ref{prop-H-choice} below). Thus, if one would like to use
Radon--Nikodym derivatives, the process $H$ should dominate both the
processes $|L|$ and $G$, which would make formulas rather cumbersome.

Sufficient conditions for the existence and uniqueness of a solution of
equation~\eqref{capital-eq-diff} will be provided in the next section. But
now let us prove a result which shows that the solution, if it exists, does
not depend on the choice of the process $H$ and the versions of the
derivatives $l$.

\begin{proposition}
\label{prop-H-choice}
Suppose $Y$ is a solution of \eqref{capital-eq-diff}, where the derivative
process $l$ is defined as in \eqref{l-def} with respect to some \cadlag\
non-decreasing predictable process $H$ such that $G\ll H$. Then for any
\cadlag\ non-decreasing predictable process $\tilde H$ such that $G\ll\tilde
H$ and any $\P\otimes \tilde H$-version of the derivative $\tilde l = d L/ d
\tilde H$, the process $Y$ also solves \eqref{capital-eq-diff} with $\tilde
l$ in place of $l$.
\end{proposition}

\begin{proof}
Let $F\colon \R^{MN} \to \R^{MN}$ denote the function which specifies the
distribution of payoffs in \eqref{capital-eq-diff}:
\begin{equation}
F(l)^{m,n} = \frac{l^{m,n}}{|l^{\cd,n}|},\label{F-def}
\end{equation}
where $F(l)^{m,n} = 0$ if $|l^{\cd,n}|=0$. As follows from
\eqref{capital-eq-integ}, we have to show that for each $m,n$
\begin{equation}
F(l)^{m,n}\cint X^n = F(\tilde l)^{m,n}\cint X^n,\label{prop-H-choice-1}
\end{equation}
where $F(l)$ denotes the process $F(l_t(\omega))$, and $F(\tilde l)$ denotes
$F(\tilde l_t(\omega))$.

One can see that if $f,f'\ge 0$ are predictable processes such that $f=f'$
$\P\otimes G$-a.s., then $f\cint X^n = f'\cint X^n$.
We have
\[
l^{m,n} = \dd{L^{m,n}}{H} =  \dd{L^{m,n}}{\tilde H} \dd{\tilde H}{H} =
\tilde l^{m,n} \dd{\tilde H}{H} \quad \text{$\P\otimes G$-a.s.},
\]
where the second equality holds in view of claim (b) of
Proposition~\ref{A-properties} from the \hyperref[appendix]{appendix}. Since $d\tilde H/d H > 0$
$\P\otimes G$-a.s. by claim (c) of Proposition~\ref{A-properties}, we have
$F(l)^{m,n}= F(\tilde l)^{m,n}$, so \eqref{prop-H-choice-1} holds, which
finishes the proof.
\end{proof}

\subsection{A sufficient condition for the existence and uniqueness of a solution of the wealth equation}
The following theorem provides a sufficient condition for the existence and
uniqueness of a solution of equation \eqref{capital-eq-diff}. Note that the
main results of our paper, formulated in Section~\ref{sec-main}, do not
require this condition to hold (they only require a unique solution to
exist), and may be valid under less strict assumptions.

\begin{theorem}
\label{th-existence}
Suppose that for each $m$ a strategy $\bL^m$ of investor $m$ satisfies the
following two conditions.
\begin{enumerate}[leftmargin=*,label=(C\arabic*)]
\item \label{th-existence-C1} There exists a $\PP\otimes\B(\R^M_+)$-measurable
function $v^{m}_t(\omega,z)$ with values in $\R^N_+$ such that for all
$\omega\in\Omega$, $y\in D$, $t\in\R_+$, $n=1,\ldots,N$,
\begin{equation}
\bL^{m,n}_t(\omega,y) = \int_0^t v^{m,n}_s(\omega,y_{s-})
\I(\inf_{u<s}y_u^m >0) d G_s(\omega)\label{C1-1}
\end{equation}
(for $s=0$, we put $y_{0-}=y_0$), and for all $\omega\in\Omega$, $z\in\R^M_+$, $t\in\R_+$ 
\begin{equation}
|v^{m}_t(\omega,z)|\Delta G_t(\omega) \le z^m.\label{C1-2}
\end{equation}

\item \label{th-existence-C2} There exist sets $\Pi^{m,n}\in\PP$,
$n=1,\ldots,N$, a non-random function $C^m\colon \R_+^M \to (0,\infty)$, and
a predictable \cadlag\ process $\delta^m>0$ such that
\begin{equation}
v^{m,n}_t(\omega,z) = 0\ \text{for all $(\omega,t)\in \Pi^{m,n}$ and $z\in
\R_+^N$},\label{C2-1}
\end{equation}
and for all $(\omega,t)\notin \Pi^{m,n}$ and $z,\tilde z,a\in\R^M_+$
such that $z^k,\tilde z^k \in  [a^k/2,2a^k]$ for all $k$, it holds that 
\begin{align}
&v^{m,n}_t(\omega,z) \ge (C^m(a)\delta_t^m(\omega))^{-1}\text{ if
$z^m>0$},\label{C2-2}\\
&v^{m,n}_t(\omega,z) \le C^m(a)\delta^m_t(\omega),
\label{C2-3}\\
&|v^{m,n}_t(\omega,z) - v^{m,n}_t(\omega,\tilde z)| \le C^m(a)\delta_t^m(\omega)
|z-\tilde z|.  \label{C2-4}
\end{align}
\end{enumerate}
Then for any vector of initial wealth $y_0\in\RMpp$ and a predictable
non-decreasing \cadlag\ process $H$ such that $G\ll H$, equation
\eqref{capital-eq-diff} has a unique solution (up to
$\P$-indis\-tin\-gui\-sha\-bility).
\end{theorem}

The proof is provided in Section~\ref{sec-proofs}. Let us comment on the conditions
imposed in the theorem.

In condition~\ref{th-existence-C1}, equation \eqref{C1-1} restricts the
class of strategies under consideration to strategies that from the whole
information of investors' past wealth use only the knowledge of the current
wealth $y_{s-}$, on which depend the ``instantaneous'' investment rates
$v_t^{m,n}$. The indicator in the integrand appears for the purpose of
ensuring that the process $Y^m$ is non-negative: if $Y_u^m$ or $Y_{u-}^{m}$
become zero for some $u$, such a strategy stops investing afterwards. For
the same reason we require \eqref{C1-2} to hold, which means that an
investor cannot spend more money than available. Note that
\ref{th-existence-C1} implies that the realization of the strategy is
absolutely continuous with respect to $G$, i.e.\ $L^m \ll G$, which is a
reasonable requirement since if a strategy does not have this property, then
it ``wastes'' money (invests in assets when the expected payoff is zero).

Condition~\ref{th-existence-C2} is needed because the proof is based on a
contraction mapping argument. Inequalities \eqref{C2-2}--\eqref{C2-3} are
analogous to similar upper and lower bounds on equation coefficients in such
proofs, while \eqref{C2-4} is a Lipschitz continuity condition. Note that it
would be too restrictive to require $v_t^{m,n}$ to be bounded away from zero
globally in inequality \eqref{C2-2}. Indeed, if asset $n$ does not yield a
payoff ``predictably'' at time $t$, it would be natural to take
$v_t^{m,n}=0$. Therefore, we relax the lower bound on $v^{m}$ by introducing
the sets $\Pi^{m,n}$ where $v^{m,n}$ may vanish.

The conditions of the theorem may look cumbersome, but it is possible to
verify that certain strategies satisfy them. In particular, in
Section~\ref{sec-candidate-optimal} we do that for a candidate optimal
strategy under mild additional assumptions.

\section{Optimal strategies}
\label{sec-optimal-strategy}

\subsection{Definition}
If a strategy profile and a vector of initial wealth are feasible, we define
the relative wealth of investor $m$ as the process
\[
r_t^m = \frac{Y_t^m}{|Y_t|},
\]
where $r_t^m(\omega) = 0$ if $|Y_t(\omega)| = 0$.

We will be interested in finding strategies for which the relative wealth of
an investor grows on average in the following sense.

\begin{definition}
For a given payoff process $X$, we call a strategy $\bL$ \emph{relative
growth optimal} for investor $m$, if for any feasible initial wealth and a
strategy profile where investor $m$ uses this strategy, it holds that $Y_t^m
> 0$ for all $t\ge 0$ and $\ln r^m$ is a submartingale.
\end{definition}

Observe that if a strategy is relative growth optimal, then also $r^m$ is a
submartingale by Jensen's inequality. Another corollary from the relative
growth optimality is that such a strategy is a \emph{survival strategy} in
the sense that the relative wealth of an investor who uses it always stays
bounded away from zero,
\begin{equation}
\inf_{t\ge 0} r_t^m >0 \text{ a.s.},\label{survival}
\end{equation}
(we use the terminology of \cite{AmirEvstigneev+13}; note that, for example,
in \cite{BlumeEasley92}, the term ``survival'' has a somewhat different
meaning). This follows from the fact that $\ln r^m$ is a non-positive
submartingale, and hence it has a finite limit $z = \lim_{t\to\infty} \ln
r_t^m$. Therefore, $\lim_{t\to\infty} r_t^m = e^z > 0$.

It is worth mentioning that the survival property \eqref{survival} also
implies that an investor who uses such a strategy achieves a not slower
asymptotic growth rate of wealth than any other investor in the market, i.e.
for any $k$
\begin{equation}
\limsup_{t\to\infty} \frac1t{\ln Y_t^m} \ge \limsup_{t\to\infty} \frac1t{\ln
Y_t^k}\ \text{a.s.}
\label{asymp-growth-opt}
\end{equation}
This property is analogous to the notion of asymptotic growth optimality in
single-investor market models (see, e.g., Section~3.10 in
\cite{KaratzasShreve98}). The validity of \eqref{asymp-growth-opt} follows
from that $\sup_{t\ge 0} |Y_t|/Y_t^m < \infty$ by \eqref{survival}, so
$\sup_{t\ge0} Y_t^k/Y_t^m < \infty$ for any $k$. Hence $ \limsup_{t\to
\infty} t^{-1}\ln({Y_t^k}/{Y_t^m}) \le 0$, from which one can
obtain~\eqref{asymp-growth-opt}.

\subsection{A candidate relative growth optimal strategy}
\label{sec-candidate-optimal}
Denote by $\nu_{\{t\}}$ the predictable random measure on $\B(\R_+^N)$
defined by
\[
\nu_{\{t\}}(\omega, A) = \nu(\omega, \{t\}\times A), \qquad A \in
\B(\R_+^N),
\]
and introduce the predictable process
\[
\bnu_t = \nu_{\{t\}}(\R_+^N).
\]
One can see that $\bnu_t$ is the conditional probability of a jump of the
process $X_t$ given the $\sigma$-algebra $\F_{t-}$ \cite[Proposition
II.1.17]{JacodShiryaev02}, i.e. $ \bnu_t = \P(\Delta X_t \neq 0 \mid
\F_{t-})$. We will always assume that a ``good'' version of the compensator
is chosen -- such that $\bnu_t(\omega) \in [0,1]$ for all $\omega,t$.

The candidate relative growth optimal strategy, which we define below,
will behave differently at points $t$ where $\bnu_t = 0$ and where $\bnu_t>0$.
To deal with them, let us partition
$\Omega\times\R_+\times(0,\infty)$ into the following three sets
belonging to $\PP\otimes\B(\R_+)$:
\begin{align*}
&\Gamma_0 = \{(\omega,t,c) : \bnu_t(\omega) = 0\},\\
&\Gamma_1 = \biggl\{(\omega,t,c) : 0<\bnu_t(\omega) < 1, \ \text{or}\
\bnu_t(\omega)=1\ \text{and}\ 
\int_{\R_+^N} \frac{c}{|x|} \nu_{\{t\}}(\omega,dx) >  1 \biggr\},\\ 
&\Gamma_2 = \biggl\{(\omega,t,c) : \bnu_t(\omega) = 1\ \ \text{and}\ \   
\int_{\R_+^N} \frac{c}{|x|} \nu_{\{t\}}(\omega,dx) \le 1\biggr\}.
\end{align*}
In the definition of the optimal strategy, the argument $c$ in the triple
$(\omega,t,c)$ will correspond to the value of the total wealth of all the
investors right before time $t$, i.e.\ $|Y_{t-}|$ (points $(\omega,t,c)$
with $c=0$ are not included in any of the sets; it will be easier to deal
with them separately). Roughly speaking, the sets $\Gamma_i$ differ in the
conditional size of possible jumps of the payoff process $X$. On $\Gamma_0$,
the conditional probability of a jump is zero. On $\Gamma_2$, only ``large''
jumps of $X$ occur (large relatively to the current total wealth), and
$\Gamma_1$ is the remaining set where both ``small'' and ``large'' jumps can
occur.

The next lemma defines an auxiliary function $\zeta$ which will be needed to
specify what proportion of wealth the optimal strategy keeps in cash.
\begin{lemma}
\label{lem-zeta-def}
For each $(\omega,t,c)\in\Gamma_1$, there exists a unique solution
$z^*(\omega,t,c)\in(0,c)$ of the equation
\begin{equation}
\int_\RNP \frac{c}{z+|x|} \nu_{\{t\}}(\omega,dx) =
1-\frac{c}{z}(1-\bnu_t(\omega)).
\label{zeta-def}
\end{equation}
The function $\zeta(\omega,t,c)$ defined on $\Omega\times\R_+\times\R_+$ by
\begin{equation}
\zeta = c \I(\Gamma_0) + z^*\I(\Gamma_1)\label{zeta-def-2}
\end{equation}
is $\PP\otimes\B(\R_+)$-measurable.
\end{lemma}
\begin{proof}
For $(\omega,t,c)\in\Gamma_1$, the left-hand side of~\eqref{zeta-def} is a
strictly decreasing continuous function of $z$, while the right-hand side is
a non-decreasing continuous function of $z$.
The existence and uniqueness of the solution $z^*$ then follows from
comparison of their values at $z=c$ and $z\to0$.

To prove the measurability of $\zeta$, consider the function $f$ defined on
$\Omega\times \R_+^3\to \R$ by
\[
f(\omega,t,c,z) = \biggl(
\int_\RNP \frac{c}{z+|x|} \nu_{\{t\}}(\omega,dx) -1 +
\frac{c}{z}(1-\bnu_t(\omega)) \biggr) \I((\omega,t,c)\in \Gamma_1) \wedge 1.
\]
Observe that $f$ is a Carath\'eodory function, i.e.
$\PP\otimes\B(\R_+)$-measurable in $(\omega,t,c)$ and continuous in $z$.
Then by Filippov's implicit function theorem (see, e.g.,
\cite[Theorem~18.17]{AliprantisBorder06}), the set-valued function
\[
\phi(\omega,t,c) = \{z\in[0,c]  : f(\omega,t,c,z)=0\}
\]
admits a measurable selector. Since $\phi$ on $\Gamma_1$ is single-valued
($\phi(\omega,t,c) = \{z^*(\omega,t,c)\}$), this implies the
$\PP\otimes\B(\R_+)$-measurability of $\zeta$.
\end{proof}

It is known that there exists a predictable process $b_t$ with values in
$\R^N_+$ and a transition kernel $K_{\omega,t}(dx)$ from
$(\Omega\times\R_+,\PP)$ to $(\R^N_+,\B(\R^N_+))$ such that up to
$\P$-indis\-tin\-gui\-sha\-bility
\begin{equation}
\label{decomp}
X_t^{c}(\omega) = b\cint  G_t(\omega), \qquad \nu(\omega,dt,dx) = K_{\omega,t}(dx) d
 G_t(\omega).
\end{equation}
Since the filtration is complete, we can assume \eqref{decomp} holds for all
$\omega,t$. Also, it will be convenient to select ``good'' versions of $b$
and $K$, which satisfy the following conditions for all $(\omega,t)$ (it is
always possible to select such versions, see, e.g., \cite[Proposition
II.2.9]{JacodShiryaev02}):
\begin{align}
&|b_t(\omega)| = 0\ \text{if}\ \Delta G_t(\omega)>0,\quad
K_{\omega,t}(\{0\}) = 0,\notag\\
&|b_t(\omega)| + \int_\RNP (1\wedge|x|) K_{\omega,t}(dx) = 1.\label{good-version}
\end{align}
Define the $\PP\otimes\B(\R_+)$-measurable function
$\hat\lambda(\omega,t,c)$ with values in $\R_+^N$:
\begin{equation}
\hat \lambda_t(0) = 0,\qquad
\hat \lambda_t(c) = \frac{b_t}{c} +  \int_{\R_+^N}
\frac{x}{\zeta_t(c)+|x|} K_t(dx) \ \ \text{for}\ \ c>0
\label{lambda-hat}
\end{equation}
(the argument $\omega$ is omitted for brevity). Now we are in a position to
introduce the strategy, which will be shown to be relative growth optimal.
When used by investor $m$, its cumulative investment process is defined by
\begin{equation}
\hat \bL_t(y) = \int_0^t y_{s-}^m \hat\lambda_s(|y_{s-}|)d G_s
\label{L-hat}
\end{equation}
(for $s=0$, put $y_{0-} = y_0$). When it is necessary to emphasize that this
strategy, as a function of $y$, depends on which investor uses it, we will
use the notation $\hat\bL_t^m(y)$.

Generally speaking, the strategy $\hat\bL$ resembles optimal strategies in
other models in evolutionary finance, as they all split investment budget
between assets proportionally to expected asset payoffs (but quantitatively
they differ in how these proportions are calculated). In the particular case
when the payoff process $X_t$ is discrete-time (as in
Section~\ref{sec-discrete-model}), we obtain the same strategy that was found
in \cite{DrokinZhitlukhin20}. Formally, the discrete-time case can be
included in the general model by taking a process $X_t$ such that $X_t =
\sum_{s=0}^{\lfloor t\rfloor} \Delta X_s$; then $b=0$ and $K_t(dx)$ is the
conditional distribution of the jump $\Delta X_t$ for integer $t$.

To conclude this section, we state a proposition which provides sufficient
conditions of feasibility of a strategy profile where one or several
investors use the strategy $\hat \bL$. It is based on
Theorem~\ref{th-existence}, but we show that the conditions of that theorem
hold automatically for $\hat\bL$ under some mild additional assumptions on
the payoff process. In particular, if these assumptions hold, then a
strategy profile where all the investors use the strategy $\hat\bL$ is feasible
(we will consider such profiles in Theorem~\ref{th-equilibrium} in the next
section).

Define the predictable process with values in $\RNP$ 
\begin{equation}
h_t = b_t + \int_\RNP \frac{x}{1+|x|} K_t(dx),\label{process-h}
\end{equation}
and define the scalar predictable process
\[
p_t = \int_\RNP \frac{\nu_{\{t\}} (dx)}{(1+|x|)^2}.
\]

\begin{proposition}
\label{pr-Lhat-conditions}
Suppose the process $(p_t \Delta G_t)^{-1}\I(\Delta G_t > 0)$ is locally
bounded and for each $n$ the process $(h_t^n)^{-1} \I(h_t^n>0)$ is locally
bounded (where $0/0=0$ for these processes). Then any strategy profile, in
which every investor uses either the strategy $\hat\bL$ or a strategy which
satisfies the conditions of Theorem~\ref{th-existence}, is feasible for any
initial wealth $y_0\in\RMpp$.
\end{proposition}

The proof is given in Section~\ref{sec-proofs}.

\section{The main results}
\label{sec-main}
The following three theorems are the main results on relative growth optimal
strategies. For convenience, we divide this section into three parts, each
containing a theorem and comments. The proofs are in
Section~\ref{sec-proofs}.

\paragraph{1.}
The first result establishes the existence of a relative growth optimal
strategy ($\hat\bL$ is such a strategy) and shows that it is, in a certain
sense, unique.

\begin{theorem}
\label{th-survival}
1. The strategy $\hat\bL$ is relative growth optimal.

2. Suppose $\bL$ is a strategy of investor $M$ such that the profile
$(\hat\bL^1,\ldots,\hat\bL^{M-1},\bL)$ and a vector of initial wealth
$y_0\in\RMpp$ are feasible and $r^M$ is a submartingale. Then $\bL_t(Y) =
\hat\bL^M_t(Y)$ for all $t\ge0$, where $Y$ is the solution of the wealth
equation for this strategy profile and initial wealth.
\end{theorem}

Let us comment on the second part of the theorem. It can be regarded as a
uniqueness result for a relative growth optimal strategy: if $M-1$ investors
use the strategy $\hat\bL$, then the remaining investor, who wants the
relative wealth to be a submartingale, has nothing to do but to act as using
the strategy $\hat\bL$ as well. Here, ``to act'' means that the realization
of the strategy of this investor, i.e.\ the cumulative investment process
$L_t(\omega) = \bL_t(\omega,Y(\omega))$ coincides (up to
$\P$-indistinguishability) with the process $\hat L_t^M(\omega)
=\hat\bL_t^M(\omega,Y(\omega))$. As a consequence, the relative wealth of
each investor will stay constant.

However, note that the strategy $\bL_t(\omega,y)$, as a function on
$\Omega\times D\times \R_+$, may be different from $\hat\bL_t^M(\omega,y)$.
Let us provide an example. Suppose there is only one asset with the
non-random payoff process $X_t = t$ and two investors with initial wealth
$y_0^1 = y_0^2 = 1$. In this case, $G_t = t$ and the strategy $\hat\bL$, if
used by investor 2, has the form
\[
\hat\bL_t(y) = \int_0^t \frac{y^2_{s-}}{y_{s-}^1 + y_{s-}^2} ds.
\]
On the other hand, consider the strategy $\bL$ for investor 2 defined as
\[
\bL_t(y) = \int_0^t \biggl(\frac13\I(y_u^1=1\text{ for all }u< s) +
\frac{y^2_{s-}}{y_{s-}^1 + y_{s-}^2} \I(y_u^1 \neq1
\text{ for some } u< s)\biggr) ds.
\]
It is not hard to see that $\bL$ is also relative growth optimal. However it
leads to a different wealth process of investor 2 compared to $\hat\bL$ if,
for example, $\bL_t^1 \equiv 0$.

\paragraph{2.}
The second result shows that the strategy $\hat\bL$ asymptotically
determines the structure of the market in the sense that if there is an
investor who uses it, then the representative strategy of all the investors
is asymptotically close to $\hat\bL$. (By the representative strategy we
call the weighted sum of the investors' strategies with their relative
wealth as the weights; see below.) Moreover, if the representative strategy
of the other investors is asymptotically different from $\hat\bL$, they will
be driven out of the market -- their relative wealth will vanish as
$t\to\infty$.

In order to state the theorem, let us introduce auxiliary processes. Suppose
a unique solution of the wealth equation exists. Let $L_t^m(\omega) =
\bL_t^m(\omega, Y(\omega))$ be the realizations of the investors'
strategies, and, as above, $l_t^m = d L_t^m/d G_t$. For each $m$, define the
predictable process $ L_t^{\s,m} = L_t^m - l^m\cint G_t, $ which is the
singular part of the Lebesgue decomposition of $L_t^m$ with respect to $G_t$
(hence the superscript ``$\s$'').

Define the proportion $\lambda_t^m$ of wealth invested in the assets by
investor $m$ as the predictable process with values in $\RNP$ and the
components
\[
\lambda_t^{m,n} = \frac{l_t^{m,n}}{Y_{t-}^m},
\]
where $0/0=0$. Note that by condition 2 of Definition~\ref{def-feasible-profile}, we have
$l_t^{m,n} = 0$ on the set $\{(\omega,t) : Y_{t-}^m(\omega) = 0\}$
($\P\otimes G$-a.s.). Introduce also the processes of ``cumulative
proportions'' of invested wealth and their singular parts:
\[
\Lambda_t^{m} = \frac{1}{Y_{-}^m} \cint L_t^{m}, \qquad
\Lambda_t^{\s,m} = \Lambda_t^{m} - \lambda^{m}\cint G_t = \frac{1}{Y_{-}^m} \cint L_t^{\s,m},
\]
which are non-decreasing, predictable, \cadlag, and with values in
$[0,+\infty]^N$.

For a set of investors $\mathbb{M}\subseteq\{1,\ldots,M\}$, let us denote
their total wealth by $Y^\mathbb{M}_t = \sum_{m\in\mathbb{M}} Y_t^m$, their
relative wealth by $r_t^\mathbb{M} = \sum_{m\in \mathbb{M}} r_{t}^m$, and
the processes associated with the realization of their representative
strategy by $L_t^\mathbb{M} = \sum_{m\in\mathbb{M}} L_t^m$, $l_t^\mathbb{M}
= d L_t^{\mathbb{M}}/ d G_t= \sum_{m\in\mathbb{M}} l_t^m$,
$L_t^{\s,\mathbb{M}} = \sum_{m\in\mathbb{M}} L_t^{\s,m}$, and
\[
\begin{aligned}
&\lambda_t^{\mathbb{M}} = \frac{l_t^{\mathbb{M},n}}{Y_{t-}^\mathbb{M}}
= \sum_{m\in\mathbb{M}}
\frac{r_{t-}^m}{r_{t-}^{\mathbb{M}}} \lambda_{t}^{m},\\
&\Lambda_t^{\mathbb{M}} =
\frac{1}{Y_{-}^\mathbb{M}} \cint L_t^{\mathbb{M}}
=  \sum_{m\in\mathbb{M}}
\frac{r_{-}^m}{r_{-}^{\mathbb{M}}} \cint \Lambda_{t}^{m},\\
&\Lambda_t^{\s,\mathbb{M}} =
\frac{1}{Y_{-}^\mathbb{M}} \cint L_t^{\s,\mathbb{M}}
=  \sum_{m\in\mathbb{M}}
\frac{r_{-}^m}{r_{-}^{\mathbb{M}}} \cint \Lambda_{t}^{\s,m}.
\end{aligned}
\]
To shorten the notation, for the set of all investors
$\mathbb{M}_1=\{1,\ldots,M\}$ we will write $\bar\lambda_t^{n} =
\lambda_t^{\mathbb{M}_1,n}$, and for the set $\mathbb{M}_2=\{2,\ldots,M\}$
write $\tilde\lambda_t^{n} = \lambda_t^{\mathbb{M}_2,n}$, and similarly for
the other processes.

\begin{theorem}
\label{th-dominance}
Suppose investor 1 uses the strategy $\hat\bL$, the other investors use
arbitrary strategies $\bL^m$, and the strategy profile
$(\hat\bL^1,\bL^2,\ldots,\bL^M)$ is feasible for some initial wealth
$y_0\in\RMpp$. Then
\begin{equation}
\|\lambda^1 -\bar\lambda\|^2 \cint G_\infty + |\bar\Lambda_\infty^\s| < \infty\
\text{a.s.},\label{th-dominance-1}
\end{equation}
and, as $t\to\infty$,
\begin{equation}
r_t^1 \to 1\ \text{a.s.\ on}\ \{\omega : \|\lambda^1 -\tilde\lambda\|^2 \cint
G_\infty(\omega) = \infty\ \text{or}\ |\tilde\Lambda_\infty^\s(\omega)| =
\infty\}.
\label{th-dominance-2}
\end{equation}
\end{theorem}

Equation \eqref{th-dominance-1} expresses the idea that the investment
proportions $\bar \lambda$ of the representative strategy of all the
investors are close to $\lambda^1 = \hat \lambda$ asymptotically in the
sense that the integral $\int_0^t \|\hat \lambda_s - \bar\lambda_s\| d G_s$
converges as $t\to\infty$ and the singular part $\bar \Lambda^\s_t$ stays
bounded. If $G_\infty = \infty$, this, roughly speaking, means that $\|\hat
\lambda_t - \bar\lambda_t\|$ is small asymptotically.

Equation \eqref{th-dominance-2} shows that the strategy $\hat\bL$ drives
other strategies out of the market if they are asymptotically different from
it. This result can be also regarded as asymptotic uniqueness of a survival
strategy: if investors $m=2,\ldots,M$ want to survive against investor 1 who
uses the strategy $\hat\bL$, they should also use (at least, collectively) a
strategy asymptotically close to $\hat\bL$.

\paragraph{3.}
Theorems~\ref{th-survival},~\ref{th-dominance} lead to the natural question:
since the strategy $\hat\bL$ is s good, what will happen if all the investors
decide to use it? Obviously, in this case their relative wealth will remain
the same. However, it is interesting to look at
the asymptotic behavior of the absolute wealth $W_t:=|Y_t|$. A priori, it is
even not obvious whether it will grow. Our third result partly answers this
question: we prove that $W$ does not decrease ``on average'' and provide a
condition for $W_t\to\infty$ as $t\to\infty$.

\begin{theorem}
\label{th-equilibrium}
Suppose all the investors use the strategy $\hat\bL$, and the initial wealth
$y_0\in\RMpp$ and the strategy profile $(\hat\bL,\ldots,\hat\bL)$ are
feasible.

Then the process $V_t:=1/W_t$ is a supermartingale and there exists the
limit $W_\infty := \lim_{t\to\infty} W_t \in(0,\infty]$ a.s. Moreover, if
$X$ is quasi-continuous (i.e. $\bnu \equiv 0$), then $\{W_\infty=\infty\} =
\{(1\wedge|x|^2 ) * \nu_\infty=\infty\}$ a.s.
\end{theorem}

If $\E |X_t| <\infty$ for all $t$, then also $\E W_t < \infty$ (since $W_t
\le |y_0| + |X_t|$), and the process $W_t$ will be a submartingale by
Jensen's inequality.
This is what we mean by the phrase that the total
wealth does not decrease on average.

It is interesting to note that if one investor uses the strategy $\hat\bL$
and the other investors use arbitrary strategies, then it does not
necessarily hold that the wealth of such an investor will grow. In
particular, it may happen that $W_t\to0$ as $t\to\infty$, which is
remarkable because an investor always has a trivial strategy which
guarantees that the wealth will not vanish -- just keep all the wealth in
cash. An example can be found in \cite{DrokinZhitlukhin20}.

Another fact worth mentioning is that, as will become clear from the proof
of the theorem, the continuous part of the payoff process $X$ does not
affect the process $W$ if all the investors use the strategy $\hat\bL$, i.e.\
$W$ will be the same for any payoff processes $X$ and $X'$ such that $X-X'$
is a continuous process. For example, if $X$ is continuous, then $W_t=W_0$
for all $t\ge 0$ even if $X$ is a strictly increasing process. In
particular, observe that in the second claim of the theorem, the continuous
part of $X$ does not enter the condition for having $W_\infty=\infty$.

\section{Proofs}
\label{sec-proofs}
\subsection{Proof of Theorem~\ref{th-existence}}
Without loss of generality we will assume that the functions $C^m$ and the
processes $\delta^m$ are the same for all the investors, since otherwise one
can take $C(a) = \max_m C^m(a)$ and $\delta_t = \max_m\delta_t^m$. Moreover,
we can assume that $\delta$ is a non-decreasing process or, otherwise, take
$\delta'_t = \sup_{s\le t} \delta_t$ ($\delta_t'$ will be finite-valued
since $\delta_t$ is predictable and \cadlag, and, hence, locally bounded;
see, e.g., VII.32 in \cite{DellacherieMeyer82}).
Proposition~\ref{prop-H-choice} implies that it is enough to prove the
existence and uniqueness of a solution for some particular choice of the
process $H$ such that $G\ll H$. We will do this for $H=G$.

We are going to construct the process $Y$ by induction on stochastic
intervals $[0,\tau_{i,j}]$ with appropriately chosen stopping times
$\tau_{i,j}$ ($i\in\{0,1,\ldots,M\}$, $j\in\mathbb{Z}_+$) such that
$\tau_{i,j}\le \tau_{i',j'}$ if $(i,j)\le(i',j')$ lexicographically (i.e.
$i<i'$, or $i=i'$ and $j\le j'$), and $\sup_{i,j} \tau_{i,j} = \infty$.
Here, ``by induction'' means that we will construct processes $Y^{i,j}$ such
that on the set $\{(\omega,t): t \le \tau_{i,j}(\omega)\}$ they satisfy
equation \eqref{capital-eq-integ} and on this set $Y^{i,j} = Y^{i',j'}$ for
any $(i',j')\ge(i,j)$. From these processes we can form the single process
$Y$ satisfying \eqref{capital-eq-integ} on the whole set $\Omega\times\R_+$.

Before providing an explicit construction, let us briefly explain the role
that $\tau_{i,j}$ will play. The stopping times $\tau_{i,0}$ for $i\ge 1$
will be the moments when the wealth of one or several investors reaches zero
``in a continuous way'' (i.e. for some $m$ we have $Y_t^m > 0$ for
$t<\tau_{i,0}$ but $Y_{\tau_{i,0}-}^m = 0$). The index~$i$ will correspond
to the $i$-th such event. Not necessarily all the investors will eventually
have zero wealth; in that case we will put $\tau_{i,j}(\omega) = \infty$ for
$i$ starting from some $i'$ and all $j$.

Between $\tau_{i,0}$ and $\tau_{i+1,0}$ we will construct a sequence of
stopping times $\tau_{i,j} \to \tau_{i+1,0}$ as $j\to\infty$, such that on
each interval $[\tau_{i,j}, \tau_{i,j+1})$ the wealth of all the investors,
who have non-zero wealth at $\tau_{i,j}$, can be bounded away from zero by
an $\F_{\tau_{i,j}}$-measurable variable. The wealth of some of those
investors may become zero at $\tau_{i,j+1}$, but only ``by a jump''. If
$\tau_{i,0}(\omega)<\infty$, it will also hold that $\tau_{i,j}(\omega) <
\infty$ for all $j$.

The reason why we need to treat differently the moments when the wealth
reaches zero in a continuous way and by a jump is that we do not assume that
the function $C(a)$ is bounded in a neighborhood of zero (this is necessary,
for example, to apply the theorem to the strategy $\hat\bL$ -- see the proof
of Proposition~\ref{pr-Lhat-conditions}).

Now we will proceed to the construction of $\tau_{i,j}$ and $Y^{i,j}$. Let
$\tau_{0,0} = 0$, and for all $t\ge 0$ put $Y_{t}^{0,0}=y_0$, where $y_0\in
\RMpp$ is the given initial wealth. Suppose $\tau_{i,j}$ and $Y^{i,j}$ are
constructed. We will now show how to construct $\tau_{i,j+1}, Y^{i,j+1}$.
For brevity, $i$ will be assumed fixed and omitted in the notation, so we will
simply write $\tau_j$, $Y^j$, while $Y^{j,m}$ will denote the $m$-th
coordinate of $Y^j$.

Let $A(\omega) = \{m : Y_{\tau_{j}}^{j,m}(\omega) > 0\}$ denote the set of
the investors who are still active (i.e.\ have positive wealth) at
$\tau_{j}$; for $\omega$ such that $\tau_j(\omega) = \infty$ we put
$A(\omega)=\emptyset$. Observe that $A$ is an $\F_{\tau_j}$-measurable
random set (since it is finite, the measurability means that $\I(m\in A)$
are $\F_{\tau_{j}}$-measurable functions for each $m$).

On the set $\{\omega : A(\omega) = \emptyset\}$, define $\tau_{j+1} =
\tau_j+1$ (with $\tau_{j+1}(\omega)=\infty$ if $\tau_j(\omega)=\infty)$, and
on the set $\Omega' = \{\omega : A(\omega) \neq
\emptyset\}$ define
\[
\gamma = (\delta_{\tau_j}+1)C(Y_{\tau_j}^j)
\]
and 
\begin{align}
\tau_{j+1} = \inf\biggl\{&t> \tau_j : |X_t-X_{\tau_j}| \ge
\frac{1}{4M\gamma^2}\;\wedge \ \min_{m\in A}Y^{j,m}_{\tau_j}, \label{tau-1}\\
&\text{or}\ G_t - G_{\tau_j} \ge
\frac{1}{2\gamma} \biggl(\frac{1}{2M}\;\wedge \min_{m\in
A}Y^{j,m}_{\tau_j}\biggr), \label{tau-2}\\
&\text{or}\
\delta_t \ge \delta_{\tau_j}+1,\ \text{or}\ t\ge \tau_j+1\biggr\}. \label{tau-3}
\end{align}
Observe that we have the strict inequality $\tau_{j+1} > \tau_j$ on
$\Omega'$ since the processes $X,G,\delta$ are \cadlag. Also, $\tau_{j+1}
\le \tau_j + 1$ by the condition in \eqref{tau-3}.

For each $\omega$ define the complete metric space
$\mathbb{E}(\omega)$ consisting of \cadlag\ functions $f\colon \R_+\to
\R_+^M$ satisfying the conditions
\begin{align}
&f_t =  Y^j_t(\omega)\ \text{for}\ t\le
  \tau_{j}(\omega), \label{E-1}
  \\
&f_t^m \in \biggl[\frac12 Y_{\tau_{j}}^{j,m}(\omega),\; 2 Y_{\tau_{j}}^{j,m}(\omega)\biggr]\
  \text{for}\ t> \tau_{j}(\omega),\ m=1,\ldots,M,\label{E-2}
\end{align}
and the metric 
\[
d(f,\tilde f) = \sup_{t\ge0} |f_t-\tilde f_t|
\]
(note that if $A(\omega)=\emptyset$, then $\mathbb{E}(\omega)$ consists of
one element).

From now on, we will assume that $\omega$ is fixed and omit it in the
notation. Consider the operator $U$ on $\mathbb{E}$, which maps a
function $f\in\mathbb{E}$ to the \cadlag\ function
$g:=U(f)\colon \R_+ \to \R_+^M$
defined by the formula
\begin{multline}
g_t^m = Y_{t\wedge\tau_{j}}^{j,m}  -
\int_0^t  |v^m_s(f_{s-})| \I(\tau_{j}<s<\tau_{j+1},\; m\in A) d G_s \\+
\int_0^t F^m(l_s(f_{s-})) \I(\tau_{j}<s<\tau_{j+1}) d X_s,
\label{operator-U}
\end{multline}
where $F\colon \R^{MN}\to\R^{MN}$ is the function defined in \eqref{F-def},
and
\begin{equation}
l_s^{m,n}(\omega,z) = v^{m,n}_s(\omega,z) \I(m\in A(\omega)).\label{U-derivative}
\end{equation}

Let us show that $U$ is a contraction mapping of $\mathbb{E}$ to itself. If
$A(\omega) = \emptyset$, this is obvious, so consider $\omega$ such that
$A(\omega)\neq\emptyset$. Suppose $f\in \mathbb{E}$, $g = U(f)$. First we
will show that $g\in \mathbb{E}$. It is clear that $g$ satisfies
\eqref{E-1}, and, if $m\notin A$, then $g^m$ satisfies \eqref{E-2}. To show
that the lower bound in \eqref{E-2} is satisfied for $m\in A$, consider the
first integral in \eqref{operator-U}. Since $f\in \mathbb{E}$, by condition
\eqref{C2-3} we have $v_s^{m,n}(f_{s-}) \le
C(Y^j_{\tau_j})\delta_s \le \gamma$, using that $\delta_s <
\delta_{\tau_j}+1$ for $s<\tau_{j+1}$. Hence the integral can be bounded
from above by $\gamma ( G_{\tau_{j+1}-} - G_{\tau_j})$, and this
quantity does not exceed $\frac12 Y_{\tau_j}^{j,m}$ by the choice of
$\tau_{j+1}$ (see \eqref{tau-2}). Therefore, $g_t^m\ge \frac12
Y_{\tau_j}^{j,m}$ for $t\ge \tau_j$.

The upper bound from \eqref{E-2} for $m\in A$ follows from that the second
integral in \eqref{operator-U} is bounded from above by $|X_{\tau_{j+1}-} -
X_{\tau_j}|$ since $F^{m,n}(l) \le 1$ and by the choice of $\tau_{j+1}$
(see \eqref{tau-1}) we have $|X_{\tau_{j+1}-} - X_{\tau_j}| \le
Y_{\tau_{j}}^{j,m}$. Thus, $g$ satisfies
conditions~\eqref{E-1}--\eqref{E-2}, so $g\in \mathbb{E}$.

Now we will show that $U$ is contracting. Consider $f,\tilde f \in
\mathbb{E}$ and $m\in A$. Then
\begin{multline*}
|U(f)^m_t - U(\tilde f)^m_t | \le \int_{(\tau_j,\tau_{j+1})} |v^m_s(f_{s-}) -
v^m_s(\tilde f_{s-})| d G_s \\+ \sum_{n=1}^N
\int_{(\tau_j,\tau_{j+1})} |F^{m,n}(l_s(f_{s-})) - F^{m,n}(l_s(\tilde f_{s-}))| d
X_s^n := \mathcal{I}_1^m+\mathcal{I}_2^m.
\end{multline*}
By conditions \eqref{C2-4} and \eqref{tau-3}, we have $|v^{m,n}_s(f_{s-}) -
v^{m,n}_s(\tilde f_{s-})| \le\gamma d(f,\tilde f)$ for $s\in(\tau_j,\tau_{j+1})$. Hence
\[
\mathcal{I}_1^m \le
\gamma d(f,\tilde f) (G_{\tau_{j+1}-} -
G_{\tau_j}) \le \frac{1}{4M} d(f,\tilde f),
\]
where the last inequality is due to \eqref{tau-2}.

To bound $\mathcal{I}_2^m$, observe that for each $n$ and $(\omega,t) \in
(\tau_j,\tau_{j+1})\setminus \Pi^{m,n}$ we have (by \eqref{C2-2})
\begin{equation}
|l_s^{\cd,n}(f_{s-})|\ge \min_{m\in A} v_s^{m,n}(f_{s-})  \ge \frac1{\gamma}, \label{l-lower-bound}
\end{equation}
and a similar inequality is true for $|l_s^{\cd,n}(\tilde f_{s-})|$. It is
straightforward to check that $F$ satisfies the property
\[
\biggl |\frac{\partial F^{m,n}}{\partial l^{p,q}}(l)\biggr| \le
\frac{1}{|l^{\cd,n}|} \qquad\text{for any }m,n,p,q.
\]
Hence, for any $l,\tilde l\in \R^{MN}_+$ such that $|l^{\cd,n}|\ge \alpha$
and $|\tilde l^{\cd,n}|\ge \alpha$
for all $n$ with some $\alpha>0$, we have $|F^{m,n}(l) - F^{m,n}(\tilde
l)| \le \alpha^{-1} | l-\tilde l|$.
From this and \eqref{l-lower-bound}, we
find that on the set $(\tau_j,\tau_{j+1})\setminus \Pi^{m,n}$
\[
|F^{m,n}(l_s(f_{s-})) - F^{m,n}(l_s(\tilde f_{s-}))| \le  \gamma
| l_s(f_{s-}) - l_s(\tilde f_{s-})| \le \gamma^2 d(f,\tilde f).
\]
On $\Pi^{m,n}$ we have 
\[
F^{m,n}(l_s(f_{s-})) - F^{m,n}(l_s(\tilde f_{s-})) = 0,
\]
and, consequently, obtain the bound
\[
\mathcal{I}_2^m \le \gamma^2 d(f,\tilde f) | X_{\tau_{j+1}-} - X_{\tau_j}|
\le \frac1{4M}d(f,\tilde f).
\]
Now we see that $U$ is a contraction mapping: $d(U(f), U(\tilde f)) \le \frac12
d(f,\tilde f)$.

As a result, $U(\omega)$ has a fixed point $f^*(\omega)$ for any $\omega$.
Observe that the operator $U$ preserves adaptedness, i.e. if $f_t(\omega)$
is a \cadlag\ adapted process with values in $\R_+^M$ and satisfies
conditions \eqref{E-1}--\eqref{E-2}, then $U(\omega, f(\omega))$ is such a
process as well. Hence $f^*$ is a \cadlag\ adapted process since it can be
obtained, for example, as the limit $U^{(n)}(Y_{t\wedge \tau_j}^j)$ as
$n\to\infty$ where $n$ stands for the $n$-times application of $U$.

Now we can define the process $Y^{j+1}$ as follows: for each $m$ put
\begin{align*}
&Y_t^{j+1,m} = f^{*,m}_t \quad \text{for}\ t < \tau_{j+1},\\
&\begin{aligned}
Y^{j+1,m}_{t} &= f^{*,m}_{\tau_{j+1}-} - 
|v_{\tau_{j+1}}^m(f^{*,m}_{\tau_{j+1}-})| \I(m\in A)\Delta G_{\tau_{j+1}} \\&+ \sum_{n=1}^N
F^{m,n}(v_{\tau_{j+1}}(f^{*,m}_{\tau_{j+1}-})) \I(m\in A)\Delta X_{\tau_{j+1}}^n\qquad\text{for}\ t \ge \tau_{j+1}
\end{aligned}
\label{X-extension}
\end{align*}
(note that $Y^{j+1}_t = Y^{j+1}_{\tau_{j+1}}$ for all $t\ge \tau_{j+1}$).
Inserting $Y^{j+1,m}$ in~\eqref{operator-U}, we obtain the equation for
$t\in [\tau_j,\tau_{j+1}]$
\begin{equation}
\label{operator-U-fixed}
Y^{j+1,m}_{t} = Y_{\tau_j}^{j,m}  -
\int_{(\tau_j,t]}   |v^m_s(Y^{j+1}_{s-})|\I(Y_{\tau_j}^{j,m}>0) d G_s 
+ \int_{(\tau_j,t]} F^m(l_s(Y^{j+1}_{s-}))  d X_s.
\end{equation}
The indicator here can be equivalently replaced by $\I(\inf_{u< s}
Y_{u}^{j+1,m}>0)$, so the first integral becomes equal to
$|\bL^m_t(Y^{j+1})| - |\bL^m_{\tau_j}(Y^{j+1})|$ by~\eqref{C1-1}. In the
second integral, on $(\tau_j,\tau_{j+1}]$ we have (as follows from
\eqref{C1-1})
\[
l^{m,n}_t(Y^{j+1}_t) =\dd{\bL_t^{m,n}(Y^{j+1})}{G_t}.
\]
Consequently, \eqref{operator-U-fixed} implies that the process $Y^{j+1}$
satisfies equation \eqref{capital-eq-integ} for
$t\le \tau_{j+1}$.

Proceeding by induction, for fixed $i$ we obtain the non-decreasing sequence
of stopping times $\tau_{i,j}$ and the processes $Y^{i,j}$. Let
$\tau_{i+1,0} = \lim_{j} \tau_{i,j}\in [0,\infty]$. On $[0,\tau_{i+1,0})$
define the process $Y^{i+1,0}$ by joining $Y^{i,j}$, i.e. for $(\omega,t)$
such that $t< \tau_{i+1,0}(\omega)$ put
\[
Y^{i+1,0}_t = Y_t^{i,0}\I(t<\tau_{i,0})+ \sum_{j=1}^\infty Y^{i,j}_t
\I(\tau_{i,j-1}\le t <\tau_{i,j}).
\]
Observe that on the set $\{\tau_{i+1,0}<\infty\}$, the limit
$Y^{i+1,0}_{\tau_{i+1,0}-}$ exists, since for $t<\tau_{i+1,0}$ the process
$Y_t^{i+1,0}$ satisfies equation \eqref{capital-eq-integ}, in which the
integral processes are non-decreasing and bounded by $X^n_{\tau_{i+1,0}}$,
and the term $|L_t^m|$ is non-decreasing and bounded by
$Y_0^m+|X_{\tau_{i+1,0}}|$. For $t\ge \tau_{i+1,0}$ put
\[
Y^{i+1,0}_{t} = Y^{i+1,0}_{\tau_{i+1,0}-} - |l^{m}| + \sum_{n=1}^N F^{m,n}(l) \Delta X^n_{\tau_{i+1,0}}
\]
with
\[
l^{m,n} = v^{m,n}_{\tau_{i+1,0}}(Y^{i+1,0}_{\tau_{i+1,0}-}) \I(\inf_{s < \tau_{i+1,0}}
Y_s^{i+1,0,m}>0) \Delta G_{\tau_{i+1,0}}
\]
(the process $Y^{i+1,0}$ stays constant after $\tau_{i+1,0}$). One can see
that now $Y^{i+i,0}$ satisfies equation \eqref{capital-eq-integ} for $t\le
\tau_{i+1,0}$. Then the proof of the existence of a solution is finished by
induction. The uniqueness follows from the uniqueness of the fixed point of
the operator $U$ on each step of induction.

\subsection{Proof of Proposition~\ref{pr-Lhat-conditions}}
As follows from Theorem~\ref{th-survival} (see also the remark after its proof on
p.~\pageref{remark-not-vanish}), if a solution of the wealth equation exists
and investor $m$ uses the strategy $\hat\bL$, then the wealth of this investor does not
vanish ($Y^m>0$ and $Y^m_- > 0$). Therefore, it will be enough to prove 
Proposition~\ref{pr-Lhat-conditions} for a strategy profile $\mathcal{L}$ in which
every investor uses either a strategy satisfying the conditions of
Theorem~\ref{th-existence}, or the strategy $\hat\bL'$ such that, when used by
investor $m$, its cumulative investment process is
\[
\hat\bL_t'(m;y) = \int_0^t y_{s-}^m \hat \lambda_s(|y_{s-}|) \I(\inf_{u<s}
y_u^m > 0)d G_s
\]
(it differs from the strategy $\hat\bL$ only by the presence of the
indicator).
In order to show that such a profile is feasible, we will verify
conditions \ref{th-existence-C1}, \ref{th-existence-C2} of
Theorem~\ref{th-existence} for $\hat\bL'(m)$.

Let $v^{m,n}_t(\omega,z) = z^m \hat\lambda_t^n(\omega,|z|)$, so that
$\hat\bL'(m)$ can be represented in the form \eqref{C1-1}. Inequality
\eqref{C1-2} is satisfied because if $\Delta G_t(\omega) > 0$ (and therefore
$\bnu_t(\omega) > 0$), then $|\hat\lambda_t(\omega,c)| = \int_\RNP
|x|(\zeta_t(\omega,c) + |x|)^{-1}\nu_{\{t\}}(\omega,dx) \le 1$ as follows
from the definition of $\hat\lambda$. Hence condition~\ref{th-existence-C1}
holds.

In order to verify condition~\ref{th-existence-C2}, consider the sets
\[
\Pi^{m,n} = \{(\omega,t) : h_t^n(\omega) = 0\}
\]
and define the function $C(a)$ by
\[
C(a) = \max\biggl(\frac{2|a|\vee 1}{a^m/2},\; 2|a|\vee 1,\;
\frac{2+32|a|}{|a|^3\wedge 1}\biggr) \ \ \text{if}\ a^m>0, \quad
C(a) = 1\ \text{if}\ a^m=0,
\]
and the process $\delta_t$ by
\[
\delta_t = \sup_{s\le t} \biggl(\max_{n}
\frac{\I(h_s^n>0)}{h_s^n}  \vee \frac{\I(\Delta
G_s > 0)}{p_s \Delta  G_s} \biggr) \vee 1.
\]
The local boundedness assumptions imply that $\delta_t$ is finite-valued.

Equality \eqref{C2-1} clearly holds.
To prove inequalities
\eqref{C2-2}--\eqref{C2-3}, consider
$z,a\in\RNP$ such that $z^k \in [a^k/2,2a^k]$ for all $k$. Suppose $z^m>0$
(and, hence, $a^m>0$). Then \eqref{C2-2} follows from that
outside the set $\Pi^{m,n}$
\[
\begin{split}
v^{m,n}_t(z) = z^m \hat \lambda_t^n(|z|) &\ge \frac{z^m}{|z|\vee 1}h_t^n \ge
\frac{a^m/2}{(2|a|\vee1) \delta_t} \ge \frac{1}{C(a)\delta_t},
\end{split}
\]
where in the first inequality we used the bound $\hat \lambda_t(c) \ge
h_t/(c\vee 1)$ for any $c>0$, which can be obtained from
\eqref{lambda-hat},~\eqref{process-h} using that $\zeta_t(c) \in[0,c]$.

To prove \eqref{C2-3}, on the set $\Gamma_0$ we can
use the estimate
\begin{equation}
\hat\lambda_t^n(c) = \frac{b_t^n}{c} + \int_\RNP \frac{x^n}{c+|x|} K_t(dx) \le
\frac{h_t^n}{c\wedge 1} \le \frac{1}{c\wedge 1}.\label{lambda_hat_existence_3}
\end{equation}
The last inequality here holds since $|h_t| \le |b_t| + \int_\RNP
(1\wedge|x|)K_t(dx)=1$ (see~\eqref{good-version}).
On the set $\Gamma_1\cup\Gamma_2$, we can use the estimate
\begin{equation}
\hat \lambda_t^n(c) = \int_\RNP \frac{x^n}{\zeta_t(c) +
|x|} K_t(dx) \le K_t(\RNP)= \frac{\bnu_t}{\Delta G_t} \le \frac{1}{\Delta G_t} \label{lambda_hat_existence_4}
\end{equation}
(note that if $(\omega,t,c)\in \Gamma_1\cup\Gamma_2$, then $\Delta G_t>0$
and $b_t=0$). Therefore, we obtain
\begin{equation}
\begin{split}
v^{m,n}_t(z) &\le z^m\max\biggl(\frac{1}{|z|\wedge 1},\; \frac{\I(\Delta
G_t>0)}{\Delta G_t}\biggr) \le (2|a|\vee1)\delta_t\le C(a)\delta_t,\label{lambda_hat_existence_1}
\end{split}
\end{equation}
so \eqref{C2-3} holds.

To prove \eqref{C2-4}, suppose $z,\tilde z,a\in\RNP$ and
$z^k,\tilde z^k \in [a^k/2,2a^k]$ for all $k$. If $ z^m=\tilde z^m =
0$, then $v_t^{m,n}(z)=v_t^{m,n}(\tilde z)=0$, so \eqref{C2-4}
holds. If $\tilde z^m=0$, but $z^m>0$, then, using
\eqref{lambda_hat_existence_1}, we obtain
\[
|v_t^{m,n}(z) - v_t^{m,n}(\tilde z) | = v_t^{m,n}(z) \le (2|a|\vee1)\delta_t
\le \frac{(2|a|\vee1)\delta_t}{a^m/2} |z-\tilde z| \le C(a) \delta_t|z-\tilde z|,
\]
where we used the inequality $|z-\tilde z| \ge z^m \ge a^m/2$. In a similar way,
\eqref{C2-4} is satisfied if $z^m=0$, but $\tilde z^m>
0$.

Let us consider the case $z^m>0$, $\tilde z^m>0$. Denote $c=|z|$, $\tilde c
= |\tilde z|$. Then
\begin{equation}
|v_t^{m,n}(z) - v_t^{m,n}(\tilde z)| \le
\hat \lambda_t^n(\tilde c)|z^m-\tilde z^m|
+ z^m|\hat\lambda_t^n(c)-\hat\lambda_t^n(\tilde c)|. \label{lambda_hat_existence_2}
\end{equation}
Using \eqref{lambda_hat_existence_3}--\eqref{lambda_hat_existence_4}, the
first term in the right-hand side can be bounded as follows:
\[
\hat \lambda_t^n(\tilde c)|z^m-\tilde z^m| \le \frac{2\delta_t}{|a|\wedge 1} |z-\tilde z|.
\]
For the second term in the right-hand side of
\eqref{lambda_hat_existence_2} we have
\[
\begin{split}
&z^m|\hat\lambda_t^n(c)-\hat\lambda^n_t(\tilde c)| \le 2a^m \biggl\{ 
\frac{|c-\tilde c|}{c\tilde c} b_t^n \\&\quad+
|c-\tilde c| \biggl(\int_\RNP \frac{x^n}{(c+|x|)(\tilde c + |x|)} K_t(dx) \biggr)
\I( \Delta G_t = 0) \\ &\quad+
|\zeta_t(c)-\zeta_t(\tilde c)| \biggl(\int_\RNP \frac{x^n}{(\zeta_t(c)+|x|)(\zeta_t(\tilde c) + |x|)} K_t(dx) \biggr)
 \I( \Delta G_t > 0)\biggr\} \\ &:= 2a^m\{A_1 + A_2 + A_3\},
\end{split}
\]
where, for brevity,  $A_i$ denote the three terms in the braces.
Using that $|c-\tilde c|\le |z-\tilde z|$ and $|b_t|\le 1$, $|h_t|\le 1$, we obtain 
\[
A_1 \le \frac{4|c-\tilde c|}{|a|^2} b_t^n  \le \frac{4|z-\tilde z|}{|a|^2} ,\\
\]
and
\[
A_2 \le \frac{|c-\tilde c|}{c(\tilde c\wedge 1)} \int_\RNP \frac{x^n}{1 + |x|}
K_t(dx)  \le \frac{4|z-\tilde z|}{|a|(|a|\wedge 1)}  h_t^n \le
\frac{4|z-\tilde z|}{|a|^2\wedge 1}.
\]
Let us bound $A_3$. Assume $c\ge\tilde c$ (hence also
$\zeta_t(c)\ge\zeta_t(\tilde c)$) and $\zeta_t(c)>0$. Then we have 
\begin{equation}
\begin{split}
A_3 &\le (\zeta_t(c)-\zeta_t(\tilde c)) \biggl(\int_\RNP
\frac{1}{\zeta_t(c)+|x|} K_t(dx)\biggr) \I(\Delta G_t>0) \\&\le
\frac{\zeta_t(c)-\zeta_t(\tilde c)}{c \Delta G_t}\I(\Delta G_t>0) \le
(\zeta_t(c) - \zeta_t(\tilde c)) \frac{2\I(\Delta G_t>0)}{|a| \Delta G_t},
\end{split}
\label{lambda_hat_existence_5}
\end{equation}
where in the second inequality we used the bound
\begin{multline*}
\int_\RNP \frac{1}{\zeta_t(c)+|x|} K_t(dx) = \frac{1}{c\Delta G_t} \int_\RNP
\frac{ c}{\zeta_t( c) + |x|} \nu_{\{t\}}(dx) \\= \frac{1}{c\Delta G_t}
\biggl(1 -\frac c{\zeta_t(c)} (1-\bnu_t)\biggr) \le \frac{1}{c\Delta G_t}.
\end{multline*}
Here the second equality follows from \eqref{zeta-def} -- notice that
$(\omega,t,c)\in \Gamma_1$ because we assume $\zeta_t(c)>0$.

Now we need to bound $\zeta_t(c)-\zeta_t(\tilde c)$ in
\eqref{lambda_hat_existence_5}. Let $\Q_t$ be the random measure on $\RNP$
defined by $\Q_t(A) = \nu_{\{t\}}(A) + (1-\bnu_t)\I(0\in A)$. Observe that
$\Q_t(\RNP)=1$. Since $(\omega,t,c)\in\Gamma_1$ and $(\omega,t,\tilde
c)\in\Gamma_1\cup\Gamma_2$, from \eqref{zeta-def} and \eqref{zeta-def-2} we
find that
\[
\int_\RNP \frac{1}{\zeta_t(c) + |x|} \Q_t(dx) = \frac 1c, \qquad
\int_\RNP \frac{1}{\zeta_t(\tilde c) + |x|} \Q_t(dx) \le \frac 1{\tilde c}.
\]
From this, we obtain
\[
\begin{split}
\frac1{\tilde c} - \frac1c &\ge (\zeta_t(c)-\zeta_t(\tilde c)) \int_\RNP
\frac{\Q_{t}(dx)}{(\zeta_t(c) + |x|)(\zeta_t(\tilde c)+|x|)} \\&\ge
(\zeta_t(c)-\zeta_t(\tilde c)) \int_\RNP
\frac{\Q_{t}(dx)}{(c + |x|)^2} \\ &\ge
(\zeta_t(c)-\zeta_t(\tilde c)) \int_\RNP
\frac{\nu_{\{t\}}(dx)}{(c + |x|)^2} \ge
\frac{\zeta_t(c)-\zeta_t(\tilde
c)}{c^2\vee 1} p_t.
\end{split}
\]
Hence, we conclude that
\begin{equation}
\zeta_t(c)-\zeta_t(\tilde c) \le \frac{(c-\tilde c)(c^2\vee 1)}{c\tilde c
p_t} \le  \frac{4|z-\tilde z|}{(|a|^2\wedge 1) p_t} .
\label{lambda_hat_existence_6}
\end{equation}
From \eqref{lambda_hat_existence_5} and \eqref{lambda_hat_existence_6}, we find
\[
A_3 \le \frac{8|z-\tilde z|\I(\Delta G_t>0)}{(|a|^3\wedge 1)p_t \Delta G_t}.
\]
This implies that inequality \eqref{C2-4} is satisfied when
$z^m>0$ and $\tilde z^m>0$:
\[
\begin{split}
|v(z)_t^{m,n} - v(\tilde z)_t^{m,n}| &\le \biggl(\frac{2}{|a|\wedge 1} + 2|a|
\biggl(\frac{4}{|a|^2} + \frac{4}{|a|^2\wedge 1} +
\frac{8}{|a|^3\wedge1}\biggr)\biggr)\delta_t |z-\tilde z| \\&\le \frac{2+32|a|}{|a|^3\wedge
1} \delta_t|z-\tilde z| \le C(a) \delta_t|z-\tilde z|.
\end{split}
\]
Thus, condition~\ref{th-existence-C2} holds, which finishes the proof.

\subsection{Proof of Theorem~\ref{th-survival}}
The key idea of the proof of the first claim of the theorem is to show that
$\ln r_t$ is a $\sigma$-submartingale by showing that its \emph{drift rate}
is non-negative. Since $\ln r_t$ is a non-positive process, it will be then
a usual submartingale \cite[Proposition~3.1]{Kallsen04}. For the reader's
convenience, let us briefly recall the related notions and known results;
details can be found in, e.g., \cite{Kallsen04}.

A scalar semimartingale $Z$ with $Z_0=0$ is called a $\sigma$-submartingale
if there exists a non-decreasing sequence of predictable sets $\Pi_n \in \PP$
such that $Z_t^{\Pi_n} := \int_0^t\I_s(\Pi_n) d Z_s$ is a submartingale for
each $n$ and $\bigcup_n\Pi_n= \Omega\times\R_+$. Suppose the triplet
$(B^h,C,\nu)$ of predictable characteristics of $Z$ with respect to a
truncation function $h(z)$ admits the representation $B^h = b^h\cint G$,
$C=c\cint G$, $\nu = K\otimes G$ with predictable processes $b_t^h,c_t$, a
transition kernel $K_t(dz)$, and a non-decreasing predictable \cadlag\
process $G_t$. Then $Z$ is a $\sigma$-submartingale if and only if
$\P\otimes G$-a.e.\ on $\Omega\times\R_+$
 \[
\int_{|z|>1} |z| K_t(dz) < \infty\quad\text{and}\quad \mathfrak{d}_t := b_t^h +
\int_{\R} (z - h(z)) K_t(dz) \ge 0
\]
(see \cite[Proposition~11.2]{KaratzasKardaras07} and
\cite[Lemma~3.1]{Kallsen04}). The predictable process $\mathfrak{d}$ is
called the drift rate of $Z$ with respect to $G$. One can see that it does
not depend on the choice of the truncation function $h$ (see
\cite[Proposition~II.2.24]{JacodShiryaev02}).

Observe that if
\[
\int_\R |z| K(dz)<\infty,
\]
then $\mathfrak{d}_t = b^0_t + \int_\R z K_t(dz)$, where $b_t^0 = b_t^h -
\int_\R h(z) K_t(dz)$ is a well-defined predictable process, which does not
depend on the choice of $h$. From this we obtain the corollary that will be
used further in the proof: if $Z$ is a non-positive semimartingale, then it
will be a submartingale if $\P\otimes G$-a.s.
\begin{equation}
\int_{z<0} z K_t(dz) >- \infty \qquad\text{and}\qquad \mathfrak{d}_t = b^0_t
+ \int_\R z K_t(dz) \ge 0.\label{supermart-cond}
\end{equation}
In particular, observe that for a non-positive semimartingale it holds that
$\int_{z>0} z K(dz)<\infty$ since $K_t(\{z : z > -Z_{t-}\}) = 0$.
As a consequence, if \eqref{supermart-cond} is satisfied, then the process
$\mathfrak{d}$ is $G$-integrable and the compensator of $Z_t$ is
\begin{equation}
A_t = \mathfrak{d}\cint G_t.\label{compensator}
\end{equation}

Let us also state one auxiliary inequality, which generalizes well-known
Gibbs' inequality, and will play an important role in the proof. Suppose
$\alpha,\beta\in\RNP$ are two vectors such that $|\alpha|,|\beta|\le 1$ and
for each $n$ it holds that if $\beta^n=0$, then also $\alpha^n=0$. Then
\begin{equation}
\alpha(\ln\alpha - \ln \beta) \ge \frac{\|\alpha-\beta\|^2}{4} + |\alpha|-|\beta|,\label{logsum}
\end{equation}
where $\alpha^n(\ln \alpha^n-\ln \beta^n) = 0$ if $\alpha^n=0$. A short
direct proof can be found in \cite[Lemma~2]{DrokinZhitlukhin20}.

Now we can proceed to the proof of the first claim of the theorem. Assume
that the strategy $\hat\bL$ is used by investor $m=1$, and 
the wealth equation has a unique solution $Y_t$. We will use the notation of Section~\ref{sec-main}
and introduce the predictable $\RNP$-valued processes $\lambda_t^1,
\Lambda_t^1$ for investor 1 and $\tilde\lambda_t, \tilde \Lambda_t, \tilde
\Lambda_t^\s$ for the other investors. To keep the notation concise, from now on
the superscript ``1'' for investor 1 will be omitted, so we will
simply write $\lambda_t, \Lambda_t$. It will be also convenient to assume
that the particular version of $\lambda_t$ is selected: $\lambda_t(\omega)
= \hat\lambda_t(\omega,W_{t-}(\omega))$ for all $(\omega,t)$ with the function
$\hat\lambda_t(c)$ defined in \eqref{lambda-hat}.

Let $W_t = Y_t + \tilde Y_t$ denote the total market wealth, and $r_t =
Y_t/W_t$ the relative wealth of investor 1. Define the predictable process
$F$ with values in $\R_+^N$ by
\[
F_t^n =  \frac{\lambda_t^{n}}{r_{t-}\lambda_t^n
+(1- r_{t-})\tilde \lambda_t^{n}},
\]
where $0/0=0$. Then $Y$ and $W$ can be written as stochastic exponents
\begin{align}
&Y =  Y_0\e\biggl(-| \Lambda| +   \frac{F}{W_-} \cint X\biggr), \notag\\
&W = W_0\e\biggl( - r_- \cint|\Lambda| -(1-r_-)\cint |\tilde \Lambda| 
+ \frac1{W_-}\cint |X|\biggr).\label{proof-W}
\end{align}
Recall that the stochastic exponent of a semimartingale $S$ is the process
$\e(S)$ which solves the equation $d\e(S)_t = \e(S)_t d S_t$ with
$\e(S)_0=1$. It is known that $\e(S) > 0$ and $\e(S)_- > 0$ if and only if
$\Delta S \neq -1$, see \cite[\S\,II.8a]{JacodShiryaev02}. From the
definition of $\hat\lambda$, one can check that $\Delta(-|\Lambda| +
({F}/{W_-})\cint X)>-1$ up to an evanescent set, hence $Y>0$ and $Y_->0$.

Let $\zeta_t(\omega)$ denote the predictable process
$\zeta_t(\omega,W_{t-}(\omega))$. As follows from the definition of
$\hat\bL$ and $\zeta$, we have $ \zeta_t = (1-|\Delta\Lambda_t|) W_{t-}$.
Let $\tilde \zeta_t = (1-|\Delta\tilde \Lambda_t|) W_{t-}$. Define the
predictable function $f(\omega,t,x)$ by
\[
f_t(x) = \ln\left(\frac{\zeta_t + F_t
x}{r_{t-}\zeta_t + (1-r_{t-})\tilde\zeta_t + |x|}\right).
\] 
Using the Dol\'eans--Dade formula, which for a process of
bounded variation $S$ takes the form $\e(S)_t = \exp(S_t^c + \sum_{u\le t}
\ln(1+\Delta S_u))$, we obtain
\[
\ln r_t = \ln r_0 + (1-r_-) \cint(|\tilde \Lambda_t^c| - |\Lambda^c_t|) +
\frac{F-1}{W_-}\cint X_t^c  +
\sum_{s\le t} f_s(\Delta X_s).
\]
For the further analysis, it will be convenient to split the process $\ln
r_t$ into several parts. Let $f_t(x) = f_t^1(x)+f_t^2(x)+f_t^3(x)$, where
\begin{align*}
&f^1_t(x) = f_t(x)\I(\Delta G_t = 0, \Delta\tilde\Lambda_t = 0),\\
&f_t^2(x) = f_t(x)\I(\Delta G_t > 0),\\
&f_t^3(x) = f_t(x) \I(\Delta G_t = 0, \Delta\tilde\Lambda_t > 0).
\end{align*}
Then 
\begin{equation}
\begin{split}
\ln r_t &= \ln r_0 +  Z_t +  \tilde Z_t
\end{split}
\label{proof-5}
\end{equation}
with the processes
\begin{align}
&
Z_t = (1-r_-) \cint(|\tilde \Lambda_t^c| - |\tilde\Lambda_t^{\s c}| -
|\Lambda^c_t|) + \frac{F-1}{W_-}\cint X_t^c + \sum_{s\le t }(f^1_s +
f^2_s)(\Delta X_s),
\label{Z}\\
&\tilde Z_t = (1-r_-)\cint |\tilde \Lambda_t^{\s c}| +
\sum_{s\le t} f_s^3(\Delta X_s),\label{Z-tilde}
\end{align}
where $\tilde \Lambda_t^{\s c} = \tilde\Lambda_t^\s - \sum_{u\le t}
\Delta\tilde\Lambda_u^\s$ is the continuous part of the singular part of the Lebesgue
decomposition of $\tilde\Lambda$ with respect to $G$.

Observe that $\I(\Delta X\neq0,\; \Delta G=0,\;
\Delta\tilde\Lambda \neq0) = 0$ since the set $\{\Delta
X\neq 0,\; \Delta G=0\}$ is totally inaccessible and the process $\tilde
\Lambda$ is predictable.
Therefore,
\begin{equation}
\sum_{s\le t} f^3_s(\Delta X_s) = \sum_{s\le t}f_s^3(0) = -\sum_{s\le t}
\ln(1- (1-r_{s-})|\Delta\tilde \Lambda_s^s|).\label{f3}
\end{equation}
From this formula and \eqref{Z-tilde}, it follows that $\tilde Z_t$ is a
non-decreasing predictable \cadlag\ process, so in order to show that $\ln
r_t$ is a $\sigma$-submartingale, it is enough to show that $Z_t$ is a
$\sigma$-submartingale.

We will make use of condition \eqref{supermart-cond}. Since the process $Z$
is of bounded variation, it is not difficult to see (from, e.g., the
canonical representation of a semimartingale) that its continuous part can be
represented as $Z_t^c = b^0\cint G_t$, where $b_0$ is the predictable
process from \eqref{supermart-cond}. From \eqref{Z}, we find
\[
b^0_t = (1-r_{t-}) (|\tilde \lambda_t| - |\lambda_t|) \I(\Delta G_t=0) + \frac{(F_t-1)b_t}{W_{t-}}.
\]
The measure of jumps $\mu^Z$ of $Z$ is such that for a function $g(\omega,t,z)$
with $g(\omega,t,0) = 0$ we have
\[
g*\mu^Z_t = g(f^1+f^2)*\mu_t + \sum_{s\le t} g(f_s^2(0)) \I(\Delta X_s=0),
\]
so its compensator can be represented in the form
$\nu^Z = K^Z d G$ with the kernel $K^Z$ such that
\[
\int_\R g_t(z) K_t^Z(dz) = \int_\RNP g_t(f_t^1(x)+f_t^2(x)) K_t(dx) +  \frac{1-\bnu_t}{\Delta G_t}g_t(f_t^2(0))
\]
(when $\Delta G_t(\omega) =0$, we have $f_t^2(\omega,x)=0$, so we treat the
last term in the right-hand side as zero). Consequently, the drift rate of
$Z$ with respect to $G$ is $ \mathfrak{d}_t = b_t^0 + \int_\R zK_t^Z(dz) =
h^1_t + h^2_t$ with the predictable processes
\begin{align}
&h^1_t =  (1-r_{t-}) (|\tilde \lambda_t| - |\lambda_t|) \I(\Delta G_t = 0) +
\frac{(F_t-1)b_t}{W_{t-}} + \int_{\R^N_+} f_t^1(x) K_t(dx) ,\notag\\
&h^2_t = \int_{\R_+^n} f_t^2(x) K_t(dx) + \frac{1-\bnu_t}{\Delta
G_t} f_t^2(0). \label{h2}
\end{align}

We need to show that $h^1,h^2\ge 0$. 
For $h^1$, using the inequality $x-1\ge \ln x$ for $x>0$, we find that
\begin{equation}
{(F_t-1)b_t} \ge 
{b_t \ln (F_t)}, \label{proof-8}
\end{equation}
where we put $b^n_t\ln(F_t^n) = 0$ if $F_t^n=0$ (notice that if $F_t^n=0$,
then $\lambda_t^n=0$, so also $b^n_t=0$). 

Introduce the set $\mathcal{X}_t(\omega) = \{x\in \RNP : x^n = 0\ \text{if}\
F_t^n(\omega) = 0,\; n=1,\ldots, N\}$. On the set $\{(\omega,t,x) : \Delta
G_t(\omega) =0,\; x\in \mathcal{X}_t(\omega)\}$, using the concavity of the
logarithm, the equality $\Delta \Lambda_t = 0$ if $\Delta G_t=0$, and the
inequality $\tilde \zeta_t\le W_{t-}$ we obtain
\begin{equation}
f_t^1(x) \ge \ln\left(\frac{W_{t-} + F_t x}{W_{t-} + |x|}\right) \ge
\frac{x\ln F_t}{W_{t-} + |x|},\label{proof-9a}
\end{equation}
where we put $x^n\ln(F_t^n) = 0$ if
$F_t^n=x^n=0$.   Denote
\begin{equation}
a_t = \int_\RNP \frac{xW_{t-}}{W_{t-} + |x|} K_t(dx) .\label{proof-9b}
\end{equation}
As follows from \eqref{lambda-hat}, we have $K_t(\omega, \RNP\setminus
\mathcal{X}_t(\omega)) = 0$.
Then from \eqref{proof-9a}--\eqref{proof-9b} we obtain
\[
\int_\RNP f_t^1(x) K_t(dx) = \int_{\mathcal{X}_t}  f_t^1(x) K_t(dx)
\ge \frac{a_t \ln F_t}{W_{t-}} \I(\Delta G_t=0).
\]
Together with \eqref{proof-8}, this implies
\[
h^1_t \ge \biggl((1-r_{t-})(|\tilde\lambda_t| - |\lambda_t|) + 
\frac{(a_t+b_t)\ln F_t}{W_{t-}}\biggr)\I(\Delta G_t=0).
\]
From \eqref{lambda-hat}, it
follows that we  have $\lambda_t =
(a_t+b_t)/{W_{t-}}$ when $\Delta G_t=0$, so on the set $\{\Delta G=0\}$  
\begin{multline*}
h_t^1 \ge (1-r_{t-})(|\tilde\lambda_t| - |\lambda_t|) + \lambda_t\ln F_t  \\=
(1-r_{t-})(|\tilde\lambda_t| - |\lambda_t|) + \lambda_t (\ln
\lambda_t -\ln(r_{t-}
\lambda_t + (1-r_{t-})\tilde \lambda_t)).
\end{multline*}
Applying inequality \eqref{logsum}, we obtain
\begin{equation}
h_t^1 \ge \frac14 (1-r_{t-})^2\|\lambda_t-\tilde\lambda_t\|^2 \I(\Delta G_t=0) \ge
0.\label{h1-positive}
\end{equation}

Let us prove that $h^2\ge 0$. Consider the set $\{\Delta G>0\}$, on which we
have
\begin{multline*}
f_t^2(x)= \ln\left( \frac{\zeta_t + F_tx }{r_{t-}\zeta_t + (1-r_{t-})\tilde
\zeta_t + |x|}\right) 
\\=
\ln\left(\frac{\zeta_t + F_tx }{\zeta_t + |x|}\right) +
\ln\left( \frac{\zeta_t + |x|}{r_{t-}\zeta_t + (1-r_{t-})\tilde
\zeta_t + |x| }\right),
\end{multline*}
and, using  the concavity of the logarithm, we find that for $x\in\mathcal X_t(\omega)$
\begin{equation}
f_t^2(x) \ge \frac{x\ln F_t}{\zeta_t + |x|} + \ln \left(\frac{\zeta_t
+ |x|}{r_{t-}\zeta_t + (1-r_{t-})\tilde \zeta_t + |x|}\right) := A_t(x) +
B_t(x).\label{AplusB}
\end{equation}
For the term $A_t(x)$, applying inequality \eqref{logsum}, we get
\[
\begin{split}
\int_\RNP A_t(x) K_t(dx) &=  \lambda_t \ln(F_t) =  \lambda_t
(\ln \lambda_t - \ln(r_{t-}\lambda_t + (1-r_{t-}) \tilde \lambda_t)) \\
&\ge \frac14 (1-r_{t-})^2 \|\lambda_t-\tilde\lambda_t\|^2 + (1-r_{t-}) (|\lambda_t| - |\tilde \lambda_t|).
\end{split}
\]
 For the term $B_t(x)$, using the
inequality $\ln x \ge 1-x^{-1}$, we obtain
\[
B_t(x) \ge \frac{(1-r_{t-})(\zeta_t-\tilde \zeta_t)}{\zeta_t+|x|}.
\]
From the definition of $\zeta_t$ (see \eqref{zeta-def}), it follows that
\[
\int_\RNP \frac 1{\zeta_t+|x|} K_t(dx)   \ge \frac{1}{W_{t-}\Delta G_t}
-\frac{1-\bnu_t}{\zeta_t \Delta G_t},
\]
so we have
\begin{equation}
\int_\RNP B_t(x) K_t(dx) \ge
(1-r_{t-}) (|\tilde\lambda_t| - |\lambda_t|)
- \frac{(1-r_{t-})(1-\bnu_t)(\zeta_t-\tilde\zeta_t)}{\zeta_t \Delta G_t},
\label{B-integ}
\end{equation}
where for the first term in the right-hand side we used that $\zeta_t -
\tilde \zeta_t =
(|\tilde\lambda_t| - |\lambda_t|) W_{t-} \Delta G_t $.
Thus, using \eqref{AplusB}--\eqref{B-integ} and that $K_t(\omega,
\RNP\setminus\mathcal X_t(\omega))=0$, we find
\begin{equation}
\label{f2-bound}
\int_\RNP f_t^2(x) K_t(dx) \ge \frac14 (1-r_{t-})^2
\|\lambda_t-\tilde\lambda_t\|^2 -
\frac{(1-r_{t-})(1-\bnu_t)(\zeta_t-\tilde\zeta_t)}{\zeta_t \Delta G_t}.
\end{equation}
Also, using again the inequality $\ln x \ge 1-x^{-1}$, we obtain
\begin{equation}
f_t^2(0) = \ln\biggl(\frac{\zeta_t}{r_{t-}\zeta_t +
(1-r_{t-})\tilde\zeta_t}\biggr) \ge \frac{(1-r_{t-})(\zeta_t - \tilde\zeta_t)}{\zeta_t}.\label{f20}
\end{equation}
Hence, from \eqref{h2}, \eqref{f2-bound}, and \eqref{f20}, we find that  
\begin{equation}
\begin{split}
\label{h2-positive}
h^2_t & \ge \frac14 (1-r_{t-})^2
\|\lambda_t-\tilde\lambda_t\|^2 \I(\Delta G_t>
0)  \ge 0.
\end{split}
\end{equation}
Thus, we have proved that $h^1,h^2\ge 0$, so $\ln r_t$ is a
submartingale, which finishes the proof of the first claim of the theorem.

To prove the second claim, suppose investors $m=1,\ldots,M-1$ use the
strategy $\hat\bL$ and investor $M$ use some strategy $\bL$. If $r_t^M$ is a
submartingale, then $\ln r_t^1$ is a supermartingale by Jensen's inequality,
and hence a martingale by the first claim of the theorem. Consequently, we
find from \eqref{proof-5} (with the same notation as above)
\[
\tilde Z_t = 0\ \text{a.s.\ for all $t\ge 0$}, \qquad h^1 + h^2 = 0\ \ 
\text{$\P\otimes G$-a.s.}
\]
The first equality implies that $L^{\s,M} = 0$, so
$L^M \ll G$. The second equality, together with \eqref{h1-positive} and
\eqref{h2-positive}, implies that $\tilde \lambda_t =
\hat\lambda_t(W_{t-})$ $\P\otimes G$-a.s., and therefore $\lambda^M_t =
\hat\lambda_t(W_{t-})$ $\P\otimes G$-a.s. Then from \eqref{L-hat} we obtain
$L^M = \hat\bL^M(Y)$, which finishes the proof.

\begin{remark*}
\label{remark-not-vanish}
As can be seen from the proof, the wealth of an investor who uses the
strategy $\hat\bL$ does not vanish ($Y^m>0$ and $Y^m_->0$) on any solution
of the wealth equation (if it exists). This fact is needed in the proof of
Proposition~\ref{pr-Lhat-conditions}.
\end{remark*}

\subsection{Proof of Theorem~\ref{th-dominance}}
We will use the same notation as in the proof of Theorem~\ref{th-survival}. Since
$\ln r_t$ is a non-positive submartingale, there exists the limit $r_\infty
= \lim_{t\to\infty} r_t$. As we have shown, $\ln r_t= \ln r_0 + Z_t + \tilde
Z_t$, where $Z_t$ is a submartingale with drift rate
\[
\mathfrak{d}_t = h^1_t+h^2_t \ge \frac14 (1-r_{t-})^2\|\lambda_t-\tilde\lambda_t\|^2
= \frac14 \|\lambda_t - \bar \lambda_t\|^2.
\]
Hence the compensator $A_t=\mathfrak{d}\cint G_t$ of $Z_t$ (see
\eqref{compensator}) satisfies the inequality
\[
A_t \ge \frac14\|\lambda - \bar \lambda\|^2\cint G_t.
\]
Since $Z_t$ is bounded from above ($Z_t\le -\ln
r_0$),
$A_t$ converges to a finite limit $A_\infty$,
so $\|\lambda - \bar\lambda\|^2 \cint G_\infty<
\infty $. Moreover, on the set $\{\|\lambda-\tilde\lambda\|^2\cint G_\infty
= \infty\}$ we necessarily have $r_\infty=1$, because otherwise we would
have $A_\infty= \infty$ on this set.

From the inequality
$\ln(1-(1-r_{s-})|\Delta\tilde\Lambda_s^\s|) \le
-(1-r_{s-})|\Delta\tilde\Lambda_s^\s|$ and \eqref{Z-tilde},~\eqref{f3}, we obtain
\[
\tilde Z_t \ge (1-r_-)\cint  |\tilde\Lambda_t^\s| = |\bar\Lambda_t^\s|.
\]
Since $\tilde Z$ converges, we have $|\bar\Lambda_\infty^\s| < \infty$, and on the set
$\{|\tilde \Lambda_{\infty}^\s| = \infty\}$ we have $r_\infty=1$.

\subsection{Proof of Theorem~\ref{th-equilibrium}}
Suppose all the investors use the strategy $\hat\bL$. By virtue of
\eqref{proof-W}, $W_t = W_0 \e(S)_t$ with the process
\[
S_t = -|\hat \lambda(W_-)| \cint G_t + \frac{1}{W_-} \cint |X_t| =
-\frac{|x|}{\zeta + |x|} * \nu_t + \sum_{s\le t}\frac{|\Delta X_s|}{W_{s-}},
\]
where $\zeta$ denotes the predictable process $\zeta_t(W_{t-})$. In
particular, the continuous part $S_t^c$ and the jumps $\Delta S_t$ are given by
\begin{align*}
&S_t^c = -\frac{|x|\I(\bnu = 0)}{W_- + |x|} * \nu_t,\\
&\Delta S_t = -\int_\RNP \frac{|x|}{\zeta_t + |x|}
\nu_{\{t\}}(dx) + \frac{|\Delta X_t|}{W_{t-}} = \frac{\zeta_t +
|\Delta X_t|}{W_{t-}} -1.
\end{align*}
From the formula $\e(S)_t = \exp(S_t^c + \sum_{u\le t}
\ln(1+\Delta S_u))$, we find $V_t = V_0 \e(U)_t$ with
the process
\begin{equation}
U_t = - S_t^c- \sum_{s\le t}\frac{\Delta S_s}{1+\Delta S_s} = -S_t^c +
\sum_{s\le t} \biggl( \frac{W_{s-}}{\zeta_s + |\Delta
X_s|}-1\biggr).\label{U}
\end{equation}
The continuous part of $U_t$ is $U_t^c = - S_t^c = b^U\cint
G_t$ with the predictable process
\[
b_t^U = \int_\RNP  \frac{|x|}{W_{t-} + |x|} K_t(dx) \I(\Delta G_t=0),
\]
and the measure of jumps $\mu^U$ acts on functions $f(\omega,t,u)$ with
$f(\omega,t,0)=0$ as
\[
f*\mu_t^U =  f\biggl(\frac{W_-}{\zeta+|x|} - 1\biggr)*\mu_t + \sum_{s\le t}
f_s\biggl(\frac{W_{s-}}{\zeta_s} - 1\biggr) \I(\Delta X_s = 0, \;\bnu_s > 0),
\]
so its compensator  $\nu^U$  is such that
\[
f*\nu^U_t = f\biggl(\frac{W_-}{\zeta + |x|}-1\biggr) *\nu_t + \sum_{s\le t}
f_s\biggl(\frac{W_{s-}}{\zeta_s} - 1\biggr) (1-\bnu_s)\I(\bnu_s>0) .
\]
In particular, $\nu^U= K^U\otimes G$ with the transition kernel $K^U$ such that
\[
\int_\R f_t(u) K^U_t(du) = \int_\RNP f_t\biggl(\frac{W_{t-}}{\zeta_t +
|x|}-1\biggr) K_t(dx) + f\biggl(\frac{W_{t-}}{\zeta_t} - 1\biggr)
\frac{(1-\bnu_t)}{\Delta G_t} \I(\Delta G_t > 0).
\]
From the definition of $\zeta$ in Lemma~\ref{lem-zeta-def}, it follows that $\int_\R |u| K^U_t(du) <
\infty$,
and hence the drift rate of $U$ with respect to $G_t$ is given by
\[
\mathfrak{d}_t^U = b_t^U + \int_\R u K_t^U(du) \le 0,
\]
where the inequality follows from that on the set $\{\Delta G = 0\}$ we have $\zeta_t =
W_{t-}$, and on the set $\{\Delta G > 0\}$ we have
\[
\int_\RNP \biggl(\frac{W_{t-}}{\zeta_t + |x|}-1\biggr) K_t(dx)  \le \biggl(1-\frac{W_{t-}}{\zeta_t} \biggr)
\frac{(1-\bnu_t)}{\Delta G_t}
\]
in view of that $K_t(dx) = (\Delta G_t)^{-1} \nu_{\{t\}}(dx)$ and the
definition of $\zeta$.

Consequently, $U_t$ is a $\sigma$-supermartingale. This implies that $V_t$ is also a
$\sigma$-super\-martin\-gale, and, hence, a usual supermartingale because it is
non-negative. In particular, it has an a.s.-limit $V_\infty =
\lim_{t\to\infty} V_t\in[0,\infty)$, and therefore $W_\infty = 1/V_\infty
\in (0,\infty]$, which proves the first claim of the theorem.

If $\bnu \equiv 0$, we have $\zeta_t = W_{t-}$ for all $t$, so equation
\eqref{U} becomes
\[
U_t = -\frac{|x|}{W_- + |x|} * (\mu_t -\nu_t),
\]
and, hence, $U_t$ is a purely discontinuous local martingale
with bounded jumps, $\Delta U_t \in (-1,0]$. Consequently, according to
Proposition 7.1 in \cite{KaratzasKardaras07}, we have $\{V_\infty = 0\} =
\{|u|^2 *\nu^U_\infty = \infty\}$ a.s., or equivalently $ \{W_\infty =
\infty\} = \{ (\frac{|x|}{W_- + |x|})^2 *\nu_\infty = \infty\}$ a.s.
From this and the existence of the limit $W_\infty$ follows the second claim
of the theorem.

\section{Appendix: Lebesgue derivatives}
\label{appendix}
In this appendix we assemble several known facts about the Lebesgue
decomposition and Lebesgue derivatives of $\sigma$-finite measures, and prove
auxiliary results for random measures generated by predictable
non-decreasing \cadlag\ processes.

\medskip
\noindent
\textbf{The Lebesgue decomposition of $\sigma$-finite measures.} Let
$(\Omega,\F)$ be a measurable space. First recall the following known
result, which can be found (in a slightly different form), e.g., in Chapter~3.2
of~\cite{Bogachev07}.

\begin{proposition}
Let $\P,\tilde\P$ be $\sigma$-finite measures on $(\Omega,\F)$. Then there
exists a measurable function $Z \ge 0$ ($\P$-a.s.\ and $\tilde\P$-a.s.) and
a set $\Gamma\in\F$ such that
\begin{equation}
\tilde \P(A) = \int_A Z d\P + \tilde \P(A\cap \Gamma) \quad \text{for any
$A\in \F$},\label{leb-decomp}
\end{equation}
and
\begin{equation}
\P(\Gamma) = 0.\label{leb-decomp-Gamma}
\end{equation}
Such $Z$ is $\P$-a.s.\ unique and $\Gamma$ is $\tilde\P$-a.s.\ unique, i.e.\
if $Z'$ and $\Gamma'$ also satisfy the above properties, then $Z=Z'$
$\P$-a.s., and $\tilde\P(\Gamma\triangle\Gamma') = 0$ (where
$\Gamma\triangle\Gamma' = \Gamma\setminus\Gamma'\cup\Gamma'\setminus\Gamma$
denotes the symmetric difference).
\end{proposition}

The function $Z$ -- the Lebesgue derivative of $\tilde\P$ with respect to
$\P$ -- is denoted in this paper by ${d\tilde\P}/{d\P}$. If $\tilde\P\ll\P$,
the Lebesgue derivative coincides with the Radon--Nikodym derivative and one
can take $\Gamma = \emptyset$. When it is necessary to emphasize that the
set $\Gamma$ is related to $\tilde\P$ and $\P$, we use the notation
$\Gamma_{\tilde \P/ \P}$.

In an explicit form, $Z$ and $\Gamma$ can be constructed as follows. Let
$\Q$ be any $\sigma$-finite measure on $(\Omega,\F)$ such that $\P\ll\Q$,
$\tilde\P\ll\Q$ (for example, $\Q=\P+\tilde\P)$. Then
\[
Z = \dd{\tilde\P}{\Q}\left(\frac{d\P}{d\Q}\right)^{-1}
\I\left(\frac{d\P}{d\Q}>0\right), \qquad \Gamma =
\left\{\omega:\frac{d\P}{d\Q}(\omega) =
0\right\},
\]
where the derivatives are in the Radon--Nikodym sense.

By approximating a measurable function with simple functions, from
\eqref{leb-decomp}, it follows that for any $\F$-measurable function $f\ge
0$
\begin{equation}
\label{A-lebesgue-integration}
\int_\Omega f d\tilde\P = \int_\Omega f\frac{d\tilde\P}{d\P}d\P +
\int_\Omega f\I(\Gamma) d\tilde\P
\end{equation}
(where the integrals may assume the value $+\infty$).

The following proposition contains facts about Lebesgue derivatives that are
used in the paper.

\begin{proposition}
\label{A-properties}
Let $\P,\tilde\P,\Q$ be $\sigma$-finite measures on $(\Omega,\F)$. Then the
following statements are true.

\begin{enumerate}[label=(\alph*),itemsep=0mm,topsep=0mm]
\item Suppose $\Q$ is representable in the form $\Q(A) = \int_A f d\P +
\int_A \tilde f d\tilde\P$, where $f,\tilde f\ge 0$ are measurable
functions, and $\tilde f = 0$ $\P$-a.s. Then
\[
\dd \P{\Q} = \frac1f \I(f>0,\tilde f=0), \quad \Gamma_{\P/\Q} = \{f=0, \tilde
f=0\}.
\]

\item If $\mathrm{R}$ is a $\sigma$-finite measure such that $\emph{R}\ll
\P$ and $\mathrm{R}\ll \Q$, then
\begin{equation}
\dd{\tilde\P}{\P} = \dd{\tilde\P}{\Q} \dd{\Q}{\P}
\quad\text{$\mathrm{R}$-a.s.}\label{A-recalc-1}
\end{equation}

\item If $\mathrm{R}$ is as in (b), then $d \Q/d \P >0$ and $d \P/d\Q >0$
$\mathrm{R}$-a.s.
\end{enumerate}
\end{proposition}
\begin{proof}
(a) is obtained by straightforward verification of
\eqref{leb-decomp}--\eqref{leb-decomp-Gamma}.

(b) Observe that for any $A\in\F$ we have
\begin{equation}
\label{A-recalc-1-proof-1}
\begin{split}
\tilde \P(A)  &= \int_A \dd{\tilde\P}{\Q} d\Q +
  \tilde\P(A\cap \Gamma_{\tilde\P/\Q})\\
&= \int_A\dd{\tilde\P}{\Q} \dd{\Q}{\P}d\P +
  \int_\Omega \I(A\cap \Gamma_{\Q/\P}) \dd{\tilde\P}{\Q} d\Q +  
  \tilde\P(A\cap \Gamma_{\tilde\P/\Q})\\
&= \int_A\dd{\tilde\P}{\Q} \dd{\Q}{\P}d\P +
  \tilde\P(A\cap(\Gamma_{\Q/\P}\cup \Gamma_{\tilde\P/\Q})),
\end{split}
\end{equation}  
where to obtain the second equality we applied
\eqref{A-lebesgue-integration}, and to obtain the third one we expressed the
second integral in the second line from the equality
\[
\tilde\P(A \cap \Gamma_{\Q/\P}) = \int_\Omega
\I(A\cap\Gamma_{\Q/\P})\dd{\tilde\P}{\Q}d\Q
+\tilde\P(A\cap\Gamma_{\Q/\P}\cap\Gamma_{\tilde\P/\Q}).
\]
Suppose for $A=\{\dd{\tilde\P}{\P} > \dd{\tilde\P}{\Q}\dd{\Q}{\P}\}$ we have
$\mathrm{R}(A)>0$. Then also $\mathrm{R}(A')>0$ for $A'=A\cap
(\Gamma_{\Q/\P}\cup \Gamma_{\tilde\P/\Q}\cup\Gamma_{\tilde\P/\P})^c$ because
$\mathrm{R}(\Gamma_{\Q/\P}) = \mathrm{R}(\Gamma_{\tilde\P/\Q}) =
\mathrm{R}(\Gamma_{\tilde\P/\P})=0$. Consequently, $\P(A')>0$. But this
leads to a contradiction between decomposition \eqref{leb-decomp} and
equality \eqref{A-recalc-1-proof-1} for $\tilde \P(A')$, since according to
them we would have
\[
\int_{A'} \dd{\tilde\P}{\P} d\P= \int_{A'}
\dd{\tilde\P}{\Q}\dd{\Q}{\P} d\P,
\]
which is impossible due to the choice of $A$. Hence $\mathrm{R}(\dd{\tilde\P}{\P} >
\dd{\tilde\P}{\Q}\dd{\Q}{\P}) = 0$. In the same way we show that $\mathrm{R}(\dd{\tilde\P}{\P} <
\dd{\tilde\P}{\Q}\dd{\Q}{\P}) = 0$.

(c) follows from \eqref{A-recalc-1} if one takes $\tilde\P = \P$.
\end{proof}

\medskip
\noindent
\textbf{The Lebesgue decomposition of non-decreasing predictable processes.} Let $(\Omega,\F,(\F_t)_{t\ge0}, \P)$ be a filtered probability space
satisfying the usual assumptions, and $\PP$ be the predictable $\sigma$-algebra
on $\Omega\times\R_+$. For a scalar  non-decreasing
\cadlag\ predictable process $G$, denote by $\P\otimes G$ the
measure on $\PP$ defined as
\begin{equation}
\P\otimes G(A) = \E(I(A)\cint G_\infty),\qquad A\in\PP.
\label{A-product-measure}
\end{equation}
Observe that $\P\otimes G$ is $\sigma$-finite on $\PP$. Indeed, this can be
shown by considering the predictable stopping times $\tau_n = \inf\{t\ge 0 :
G_t \ge n\}$.
The stochastic intervals $A_n = [0,\tau_n) := \{(\omega,t) :
t<\tau_n(\omega)\}$ are predictable, i.e.\ $A_n\in\PP$, while $\P\otimes
G(A_n)\le n$ and $\bigcup_n A_n = \Omega\times \R_+$.

\begin{proposition}
\label{A-pred-leb}
(a) For any scalar non-decreasing \cadlag\ predictable processes $G,\tilde
G$ there exists a predictable process $\xi\ge 0$ and a set $\Gamma\in\PP$
such that up to $\P$-indis\-tin\-gui\-sha\-bility
\begin{equation}
\tilde G = \tilde G_0 + \xi\cint G + \I(\Gamma)\cint \tilde G \quad\text{and}\quad \I(\Gamma)\cint G =
0.\label{A-predictable-lebesgue}
\end{equation}
(b) A predictable process $\xi\ge0$ and a set $\Gamma\in\PP$ satisfy
\eqref{A-predictable-lebesgue} if and only if $\xi$ is a version of the
Lebesgue derivative $d(\P\otimes \tilde G)/d(\P\otimes G)$ and $\Gamma$ is
the corresponding set from the Lebesgue decomposition.
\end{proposition}

We denote any $\P\otimes G$-version of such a process $\xi$ by $d{\tilde
G}/d{G}$ or $d{\tilde G_t}/d{G_t}$, and call it a predictable Lebesgue
derivative of $\tilde G$ with respect to $G$. When it is necessary to
emphasize that the set $\Gamma$ is related to $\tilde G$ and $G$, we use the
notation $\Gamma_{\tilde G/ G}$.

\begin{proof}
Without loss of generality  assume $\tilde G_0=0$.

(a) Let $\xi=d(\P\otimes\tilde G)/d(\P\otimes G)$ and
$\Gamma$ be the corresponding set from the Lebesgue decomposition.
Define the process
\[
\tilde G' = \xi\cint G + \I(\Gamma)\cint \tilde G.
\]
We have to show that $\tilde G' = \tilde G$. Since $\tilde G'$ and $\tilde
G$ are \cadlag, it is enough to show that $\tilde G'_t = \tilde G_t$ a.s.\
for any $t\ge 0$, and this is equivalent to that
\begin{equation}
\E(\tilde G'_t \I(B)) = \E(\tilde G_t\I(B)) \quad\text{for any $B\in \F_t$}.
\label{A-pred-leb-1}
\end{equation}
Let $M$ be the bounded \cadlag\ martingale such that $M_u = \E(\I(B)\mid
\F_u)$. We have
\[
\E (\tilde G'_t \I(B)) = \E(\tilde G'_t M_t) = \E(M_- \cint \tilde G'_t),
\]
and, similarly,
\begin{equation}
\label{A-pred-leb-3}
\E (\tilde G_t \I(B))= \E(M_- \cint \tilde G_t),
\end{equation}
where we used the following fact: if $A_t$ is a non-decreasing \cadlag\
predictable process and $M_t$ is a bounded \cadlag\ martingale, then for any
stopping time $\tau$ we have $\E(M_\tau A_\tau) = \E(M_-\cint A_\tau)$. This
result is proved in \cite[Lemma~I.3.12]{JacodShiryaev02} in the case $\E
A_\infty < \infty$, from which our case follows by a localization procedure.

Finally, from the definition of $\tilde G'$ and the Lebesgue decomposition
of the measure $\P\otimes\tilde G$, it follows that the measures $\P\otimes
\tilde G$ and $\P\otimes \tilde G'$ coincide. Hence, for any non-negative
$\PP$-measurable function $f$ we have $\E (f\cint \tilde G_t) = \E
(f\cint\tilde G_t')$, which finishes the proof by
\eqref{A-pred-leb-1}--\eqref{A-pred-leb-3}.

(b) In view of the construction in (a), it only remains to show that if
$\xi,\Gamma$ satisfy \eqref{A-predictable-lebesgue}, then $\xi$ is the
Lebesgue derivative and $\Gamma$ is the corresponding predictable set. This
follows from straightforward verification of properties
\eqref{leb-decomp}--\eqref{leb-decomp-Gamma}.
\end{proof}

\phantomsection
\addcontentsline{toc}{chapter}{\refname}
\bibliographystyle{abbrvnat}
\bibliography{continuous-time-game}
\end{document}